\documentclass{siamonline0316}
\usepackage{ifthen}
\newcommand{\techRep}{true} 
\newcommand{\iftechrep}{\ifthenelse{\equal{\techRep}{true}}}
\usepackage{amssymb}
\usepackage{amsmath}
\usepackage{xspace}
\usepackage{tikz}
\usepackage{kbordermatrix}
\usetikzlibrary{arrows, automata, positioning, spy, 3d, calc, matrix}
\usepackage{enumitem}
 
\usepackage{kbordermatrix}

\renewcommand{\kbldelim}{(}
\renewcommand{\kbrdelim}{)}

\newcommand{\newextmathcommand}[2]{%
    \newcommand{#1}{\ensuremath{#2}\xspace}
}

\newextmathcommand{\N}{\mathbb{N}}
\newextmathcommand{\Z}{\mathbb{Z}}
\newextmathcommand{\Q}{\mathbb{Q}}
\newextmathcommand{\R}{\mathbb{R}}

\newcommand{\vone}{\mathbf{1}}%

\newcommand{\eps}{\varepsilon}
\renewcommand{\epsilon}{\varepsilon}

\newcommand{\NP}{\textsc{NP}}

\newcommand{\PSPACE}{\textsc{PSPACE}}

\newcommand{\rk}{\mathrm{rank}}%
\newcommand{\rkp}{\mathrm{rank}_{+}}%

\newcommand{\sset}{\subseteq}

\RequirePackage{tikz}
\usetikzlibrary{arrows}

\newcommand{\objectThreeD}[4]{
\pgfmathsetmacro{\phiy}{#1} 
\pgfmathsetmacro{\phix}{#2} 
\pgfmathsetmacro{\phiz}{#3} 
\pgfmathsetmacro{\xa}{cos(\phiz)*cos(\phiy)-sin(\phiz)*cos(\phix)*sin(\phiy)}
\pgfmathsetmacro{\xb}{-sin(\phix)*sin(\phiy)}
\pgfmathsetmacro{\ya}{-sin(\phiz)*sin(\phix)}
\pgfmathsetmacro{\yb}{cos(\phix)}
\pgfmathsetmacro{\za}{-cos(\phiz)*sin(\phiy)-sin(\phiz)*cos(\phix)*cos(\phiy)}
\pgfmathsetmacro{\zb}{-sin(\phix)*cos(\phiy)}
\begin{tikzpicture}[x  = {(\xa cm,\xb cm)},
                    y  = {(\ya cm,\yb cm)},
                    z  = {(\za cm,\zb cm)},
                    scale = #4,
                    dot/.style={circle,fill=black,minimum size=4pt,inner sep=0pt,outer sep=-1pt},
]
\draw[line join=round,thick,fill=brown!90] (0,0,0) coordinate (O1) -- (1,0,0) coordinate (O2) -- (1,1/2,0) coordinate (O3) -- (0,1,0) coordinate (O4) -- cycle;
\draw[line join=round,thick,fill=blue!50]  (O1) -- (O2) -- (9/4,0,1/2) coordinate (O5) -- (0,0,8/7) coordinate (O6) -- cycle;

\draw[thick,->, >=angle 60] (O1) -- (O2) node[below] {$x$};
\draw[thick,->, >=angle 45] (O1) -- (O4) node[left] {$y$};
\draw[thick,->, >=angle 90] (O1) -- (O6) node[left] {$z$};

\coordinate (I1) at (3/4,1/8,0);
\coordinate (I2) at (3/4,1/2,0);
\coordinate (I3) at (3/11,17/22,0);
\coordinate (I4) at (2,0,1/2);
\coordinate (I5) at (1/2,0,3/4);
\coordinate (I6) at (1/6,0,7/12);
\foreach \x in {1,2,...,6} \node[dot] at (I\x) {};

\coordinate (M1) at ({2-sqrt(2)}, 0,0);
\coordinate (M2) at  (intersection cs:
       first line={(M1)--(I1)},
       second line={(O2)--(O3)});
\coordinate (M3) at  (intersection cs:
       first line={(M2)--(I2)},
       second line={(O3)--(O4)});
\coordinate (M4) at  (intersection cs:
       first line={(M1)--(I4)},
       second line={(O5)--(O6)});
\coordinate (M5) at  (intersection cs:
       first line={(M4)--(I5)},
       second line={(O1)--(O6)});

\draw (M1) -- (M2) -- (M3) -- cycle;
\draw (M1) -- (M4) -- (M5) -- cycle;

\draw[line join=round,thick,fill=yellow!80]  (O2) -- (O3) -- (O5) -- cycle;
\draw[line join=round,thick,fill=red!40]  (O1) -- (O4) -- (O6) -- cycle;
\draw[line join=round,thick,fill=green!90,opacity=0.4, draw opacity=1]  (O3) -- (O4) -- (O5) -- cycle;
\draw[line join=round,thick,fill=orange!60,opacity=0.35, draw opacity=1]  (O4) -- (O5) -- (O6) -- cycle;
\end{tikzpicture}
}

\newsiamthm{claim}{Claim}
\newsiamremark{remark}{Remark}

\newcommand{\TheTitle}{Nonnegative Matrix Factorization Requires Irrationality}
\newcommand{\TheAuthors}{D. Chistikov, S. Kiefer, I. Maru\v{s}i\'{c}, M. Shirmohammadi, and J. Worrell}

\headers{\TheTitle}{\TheAuthors}

\ifpdf
\hypersetup{
  pdftitle={\TheTitle},
  pdfauthor={\TheAuthors}
}
\fi

\title{{\TheTitle}%
\footnote{%
Extended abstracts of parts of the paper previously appeared
in conference proceedings as~\cite{16CKMSW-ICALP} and~\cite{CKMSW17soda}.
}}
\author{
    Dmitry Chistikov%
    \footnote{Max Planck Institute for Software Systems (MPI-SWS), Germany}
    \thanks{University of Oxford, UK (\email{FirstName.LastName@cs.ox.ac.uk})}%
    \and
    Stefan Kiefer\footnotemark[3]
    \and
    Ines Maru\v{s}i\'{c}\footnotemark[3]
    \and
    Mahsa Shirmohammadi\footnotemark[3]
    \and
    James Worrell\footnotemark[3]
}

\begin{document}

\maketitle

\begin{abstract}
 Nonnegative matrix factorization (NMF) is the problem of decomposing
  a given nonnegative $n \times m$ matrix~$M$ into a product of a
  nonnegative $n \times d$ matrix~$W$ and a nonnegative
  $d \times m$ matrix~$H$. A longstanding open question, posed by Cohen and Rothblum in~1993, is whether a rational matrix~$M$ always has an NMF of minimal inner dimension $d$ whose factors $W$ and $H$ are also rational.  We answer this question negatively, by exhibiting a 
  matrix for which $W$ and~$H$ require irrational entries.
 \end{abstract}



\section{Introduction}\label{sec-intro}
Nonnegative matrix factorization (NMF) is the task of factoring a
matrix of nonnegative real numbers~$M$ (henceforth a
\emph{nonnegative} matrix) as a product $M= W \cdot H$ such that the
matrices~$W$ and~$H$ are also nonnegative.  As well as being a natural
problem in its own right, see Thomas~\cite{Thomas74} and Cohen and Rothblum~\cite{CohenRothblum93}, NMF has found many applications in various
domains, including machine learning, combinatorics, and communication
complexity; see, e.g.,~\cite{%
DBLP:conf/focs/AroraGM12,Gillis14,LawtonS71,Moitra16,Venkatasubramanian13,DBLP:journals/jcss/Yannakakis91}
and the references therein.

For an NMF $M = W \cdot H$, the number of columns in $W$ is called the
\emph{inner dimension}.
The smallest inner dimension of any NMF of~$M$ is called the \emph{nonnegative
rank (over the reals)} of~$M$; 
an early reference is the paper by Gregory and Pullman~\cite{GregoryP}.
Similarly, the \emph{nonnegative rank of~$M$ over the rationals}
is defined as the smallest inner dimension of an NMF $M = W \cdot H$
with matrices $W, H$ that have only \emph{rational} entries.
Cohen and Rothblum~\cite{CohenRothblum93} posed the following problem in~1993:
\begin{quote}
``%
\textsc{Problem.}  Show that the nonnegative ranks of a rational matrix over the
reals and over the rationals coincide, or provide an example where the two
ranks are different.%
''
\end{quote} 
In this paper, we solve the above problem by providing an example of a rational matrix~$M$
that has different nonnegative ranks over \R and over \Q.%

\subsection*{Discussion}


In the last few years, there has been progress towards resolving the
Cohen--Rothblum problem.  It was already known to Cohen and
Rothblum~\cite{CohenRothblum93} that the nonnegative ranks over \R and \Q
coincide for matrices of rank at most~$2$.  (Note that the usual ranks
over~\R and~\Q coincide for all rational matrices.)  In 2015, Kubjas
et al.~\cite[Corollary~4.6]{kubjas2015fixed} extended this result to
matrices of nonnegative rank (over \R) at most~$3$.  On the other
hand, Shitov~\cite{ShitovNonRankSubField15} proved that the
nonnegative rank of a matrix can indeed depend on the underlying
field: he exhibited a nonnegative matrix with irrational entries whose
nonnegative rank over a subfield of~\R is different from its
nonnegative rank over~\R. Independently and concurrently with our work, Shitov~\cite{ShitovCR} also proposes a rational matrix whose nonnegative ranks over \R and~\Q are different.

In the present paper, in order to find a rational matrix
that has different nonnegative ranks over \R and~\Q,
we proceed in two steps.
As the first step, we study \emph{restricted}
NMFs~\cite{gillis2012geometric}, that is, those
factorizations~$M' = W \cdot H'$ of a given matrix~$M'$ in which the
columns of~$W$ span the same vector space as the columns of~$M'$.
We find irrationality in this setting,
constructing a rational matrix $M'$ that has a unique
(and irrational) restricted NMF $M' = W \cdot H'$ of inner dimension~$5$;
uniqueness here is understood
up to permutation and rescaling of columns of~$W$.
As the second step, we transfer this irrationality to our main setting:
we construct, based on the matrix $M'$, another matrix $M$ that has
a unique (and irrational) NMF $M = W \cdot H$ of inner dimension~$5$.

For the first step, it
has long been known~\cite{CohenRothblum93} that NMF can be interpreted
geometrically as finding a set of vectors (columns of~$W$) inside a
unit simplex whose convex hull contains a given set of points (columns
of~$M$).
It has recently been shown
by Gillis and Glineur~\cite{gillis2012geometric} (see also \cite{16CKMSW-ICALP})
that \emph{restricted} NMFs are in one-to-one
correspondence with \emph{nested polytopes}:
a matrix $M'$ corresponds to a pair of full-dimensional
polytopes $\mathcal R \sset \mathcal P$,
and a restricted NMF of $M'$ to a polytope 
$\mathcal Q$ \emph{nested}
in between:
$\mathcal R \sset \mathcal Q \sset \mathcal P$.
In this paper we find a pair of $3$-dimensional
polytopes $\mathcal R \sset \mathcal P$ with rational vertices
such that there is only one $5$-vertex polytope $\mathcal Q$
with $\mathcal R \sset \mathcal Q \sset \mathcal P$,
and the vertices of this polytope $\mathcal Q$ have irrational
coordinates:
$\mathcal R$ and $\mathcal P$ are chosen so as
to impose quadratic constraints on the coordinates
of the vertices of $\mathcal Q$.
This gives us a rational matrix $M'$ that has a unique
(and irrational) \emph{restricted}
NMF $M' = W \cdot H'$ of inner dimension~$5$.

For the second step,
if we knew that the factorization $M' = W \cdot H'$
were unique among \emph{all} NMFs of the same inner dimension,
\color{black}
we would be done.
This, however, requires ruling out several classes of other
hypothetical (non-restricted) factorizations of the matrix.


Towards this goal, one might want to take advantage of
the work on uniqueness properties of NMF,
studied, for instance, by Thomas~\cite{Thomas74},
Laurberg et al.~\cite{LaurbergCPHJ08}, and
Gillis~\cite{Gillis12},
or on the topology of the set of all NMFs
(see Mond et~al.~\cite{MondSvS03}).
Here we pursue a different strategy.
We show that for a larger matrix
$M = \begin{pmatrix} M^{\prime} & W_\eps \end{pmatrix}$, where~%
$W_\eps$ is a nonnegative rational matrix which is entry-wise close to
$W$, the only NMF (restricted or otherwise) of the same inner
dimension has the same left factor~$W$---thus extending the uniqueness
property to the \emph{non-restricted} setting.

The guiding idea behind our extending~$M'$ to~$M$ is that by including all 
columns of~$W$ into the set of columns of~$M$ we can exclude certain 
``undesirable'' factorizations, thereby ensuring that $M$ has no rational NMF. 
We show this by a constraint propagation argument.
Unfortunately for this construction, the matrix~$W$ itself has irrational 
entries.  However, we show that we can instead take any nonnegative 
\emph{rational} matrix~$W_\eps$ within a sufficiently small neighbourhood 
of $W$, and the undesirable factorizations will still be excluded.  In the 
text we describe such a neighbourhood explicitly
and pick a specific rational 
matrix~$W_\eps$ from it, thus obtaining the matrix~$M$
of the above form and proving the main
result of the paper.

Conceptually, the existence of a suitable matrix $W_\eps$ can be
understood in terms of \emph{upper semi-continuity} of the nonnegative 
rank over \R, proved by Bocci et~al.~\cite{BocciCR11}.
By this property, if a matrix $M$ has
nonnegative rank~$r$ over \R, then all nonnegative matrices in a sufficiently
small neighborhood of $M$ have nonnegative rank~$r$ or greater (over \R).
In the same manner,
our proof extends the non-existence of undesirable
factorizations from the matrix $W$ to $W_\eps$.


From the computational perspective, nonnegative rank (over~\R as well as over~\Q)
is a nontrivial quantity to compute.
The usual rank of a matrix $M$ is greater than or equal to~$r$ if and only if
$M$ has an $r \times r$ submatrix of rank~$r$.
The same property does not hold for nonnegative rank.  This
follows from a construction by Moitra~\cite{Moitra16} of a family
of matrices, indexed by $r,n \in \mathbb{N}$, respectively
having size $3rn \times 3rn$ and nonnegative rank at least $4r$,
but no $(n-1)\times 3rn$ submatrix of nonnegative rank greater
than $3r$.
%
A strengthening of this result can be found in
Eggermont et~al.~\cite{EggermontHK14}; this paper, in fact, studies
the set of matrices of nonnegative rank at most~$3$ and looks into the
properties of the boundary of this set.

Deciding whether a given matrix has nonnegative rank at most~$r$ is a
computationally hard problem, known to be \NP-hard due to a result by
Vavasis~\cite{Vavasis09}.
The problem is
easily seen to be reducible to the decision problem for the existential
theory of real closed fields and therefore belongs to
\PSPACE\ (see, e.g.,~\cite{Can88}).
Beyond this generic
upper bound, the problem has been attacked from many different angles.
Here we highlight the results of
Arora et~al.~\cite{AroraGKM16}, who identified
several variants of the problem that are efficiently solvable,
and Moitra~\cite{Moitra16}, who found semi-algebraic descriptions
of the sets of matrices of nonnegative rank at most~$r$ in which
the number of variables is~$O(r^2)$.
%
%
However, it remains an open question~\cite{Vavasis09}
whether or not the set $\{ (M, r) \colon
\text{the nonnegative rank of $M$ is~$\le r$} \}$
belongs to \NP; our solution to the Cohen--Rothblum
problem does not exclude either possibility
(even though it does rule out a hypothetical ``simple'' argument
for membership in \NP, wherein a certificate is an NMF with rational
entries of small bit size).

%
%
%



\section{Preliminaries}
\label{sec-prelims-NMF}
For any ordered field $\mathbb{F}$, we denote by $\mathbb{F}_{+}$ the
set of all its nonnegative elements. For any vector $v$, we write
$v_{i}$ for its $i$th entry. A vector of real numbers $v$ is
called \emph{pseudo-stochastic} if its entries sum up to one. A
pseudo-stochastic vector $v$ is called \emph{stochastic} if its
entries are
nonnegative. 

For any matrix $M$, we write $M_{i, :}$ for its $i$th row, $M_{:, j}$
for its $j$th column, and $M_{i, j}$ for its $(i,j)$th entry.
A matrix is called \emph{nonnegative} (resp., \emph{zero} or \emph{rational}) if so are all its entries.
A nonnegative matrix is \emph{stochastic} if its columns are stochastic. 

\subsection{Nonnegative Rank}\label{subsec-nonnegRank}
Let $\mathbb{F}$ be an ordered field, such as the reals $\mathbb{R}$ or the rationals~$\mathbb{Q}$. Given a nonnegative matrix $M \in \mathbb{F}_+^{n \times m}$, a \emph{nonnegative matrix factorization (NMF) over $\mathbb{F}$} of $M$ is any representation of the form $M = W \cdot H$ where $W \in \mathbb{F}_+^{n \times d}$ and $H \in \mathbb{F}_+^{d \times m}$ are nonnegative matrices.
We refer to~$d$ as the \emph{inner dimension} of the NMF, and hence refer to NMF $M = W \cdot H$ as being \emph{$d$-dimensional}.
The \emph{nonnegative rank over $\mathbb{F}$} of $M$ is the smallest
nonnegative integer $d$ such that there exists a $d$-dimensional NMF
over $\mathbb{F}$ of $M$. We may equivalently characterize~\cite{CohenRothblum93} the nonnegative rank over
$\mathbb{F}$ of $M$ as the smallest number of \mbox{rank-$1$} matrices in
$\mathbb{F}_{+}^{n \times m}$ such that $M$ is equal to their sum.
The nonnegative rank over~$\mathbb{R}$ will henceforth simply be called nonnegative rank. 
For any nonnegative matrix $M \in \mathbb{R}_{+}^{n \times m}$, it is easy to see that 
$
\rk (M) \le \rkp (M) \le \min\{n, m\} 
$,
where $\rk(M)$ and $\rkp(M)$ denote the rank and the nonnegative rank, respectively.

Given a nonzero matrix $M \in \mathbb{F}_{+}^{n \times m}$, by
removing the zero columns of $M$ and dividing each remaining column by
the sum of its elements, we obtain a stochastic matrix with equal
nonnegative
rank. 
Similarly, if $M = W \cdot H$ then after removing the zero columns
in~$W$ and multiplying with a suitable diagonal matrix~$D$, we get
$M = W \cdot H = W D \cdot D^{-1} H$ where~$W D$ is stochastic.  If
$M$ is stochastic then (writing $\vone$ for a row vector of ones) we
have 
\[
\vone = \vone M = \vone W D \cdot D^{-1} H = \vone D^{-1} H,
\]
hence $D^{-1} H$ is stochastic as well.  Thus, without loss of
generality one can consider NMFs $M=W\cdot H$ in which $M$, $W$, and
$H$ are all stochastic matrices~\cite[Theorem 3.2]{CohenRothblum93}.
In such a case, we will call the factorization $M = W \cdot H$
\emph{stochastic}.



%

\subsection{Nested Polygons in the Plane}
\label{sec:NPP}
In this paper all polygons are assumed to be convex.  Given two
polygons in the plane,
$\mathcal{R}\subseteq\mathcal{P}\subseteq\mathbb{R}^2$, a polygon
$\mathcal{Q}$ is said to be \emph{nested between~$\mathcal{R}$
  and~$\mathcal{P}$} if
$\mathcal{R}\subseteq \mathcal{Q}\subseteq \mathcal{P}$.  Such a
polygon is said to be \emph{minimal} if it has the minimum number of
vertices among all polygons nested between $\mathcal{R}$
and~$\mathcal{P}$.  In this section we recall
from~\cite{DBLP:journals/iandc/YapABO89} a standardized form for
minimal nested polygons, which will play an important role in the
subsequent development.

Fix two polygons $\mathcal{R}$ and $\mathcal{P}$, with
$\mathcal{R}\subseteq\mathcal{P}$.  A \emph{supporting line segment}
is a directed line segment $u v$ such that, first, the endpoints $u$ and~$v$ lie on the
boundary of the outer polygon~$\mathcal{P}$ and, second,
the inner polygon $\mathcal{R}$ touches~$u v$ and lies to the left of~$u v$.
A nested polygon with vertices
$v_1,\ldots,v_k$, listed in anti-clockwise order, is said to be
\emph{supporting} if the directed line segments $v_1 v_2, v_2 v_3,
\ldots, v_{k-1} v_k$ are all supporting. (Note that the directed line segment $v_k v_1$ need not be supporting.)
Such a polygon is uniquely determined by the vertex $v_1$
(see~\cite[Section 2]{DBLP:journals/iandc/YapABO89}) and is henceforth
denoted by $\mathcal{S}_{v_1}$.  It is shown
in~\cite{DBLP:journals/iandc/YapABO89} that some supporting polygon is
also minimal.  More specifically, from~\cite[Lemma
4]{DBLP:journals/iandc/YapABO89} we have:
\begin{lemma}
Consider a minimal nested polygon with vertices
  $v_1,\ldots,v_k$, listed in anti-clockwise order, where $v_1$ lies
  on the boundary of $\mathcal{P}$.  The supporting polygon
  $\mathcal{S}_{v_1}$ is also minimal.
\label{lem:NPP}
\end{lemma}

We will need the following elementary fact of linear algebra in
connection with subsequent applications of Lemma~\ref{lem:NPP}.
Let $v_1=(x_1,y_1)$, $v_2=(x_2,y_2)$, and $v_3=(x_3,y_3)$ be three distinct
points in the plane, and consider the determinant
\[ \Delta = \begin{vmatrix} x_1 & y_1 & 1 \\
                   x_2 & y_2 & 1 \\
                   x_3 & y_3 & 1 
                 \end{vmatrix} \, . \] Then $\Delta=0$ if and only if
               $v_1$, $v_2$, and $v_3$ belong to some common line, and $\Delta > 0$
               if and only if the list of vertices $v_1,v_2,v_3$
               describes a triangle with anti-clockwise orientation.

\section{Main Result} 
\label{sec:main}
We show that the nonnegative ranks over $\R$ and~$\Q$ are, in general, different.
\newcommand{\dis}{1mm}%
\begin{theorem}\label{thm:matrixWithDifferentRanks}
Let $M = \begin{pmatrix} M^{\prime} & W_\eps \end{pmatrix} \in \Q_+^{6 \times 11}$ where:
\[
M^{\prime} = 
\begin{pmatrix}
\frac{5}{44} & \frac{5}{11} & \frac{85}{121} & 0 & 0  & 0  \\[\dis]
0    & 0    & 0   & \frac{2}{11} & \frac{3}{11} & \frac{7}{33}  \\[\dis]
\frac{1}{11} & \frac{1}{44} & \frac{2}{121} & \frac{1}{44} & \frac{15}{88} & \frac{17}{88} \\[\dis]
\frac{1}{44} & \frac{1}{44} & \frac{8}{121}  & \frac{1}{44} & \frac{19}{88} & \frac{5}{24} \\[\dis]
\frac{3}{11} & \frac{3}{11} & \frac{12}{121} & \frac{8}{11} & \frac{2}{11}  & \frac{2}{33} \\[\dis]
\frac{1}{2} & \frac{5}{22} & \frac{14}{121} & \frac{1}{22} & \frac{7}{44} & \frac{43}{132}
\end{pmatrix}
\in \Q_+^{6 \times 6},
\]
\[
W_\eps =
\begin{pmatrix}
0                     & \frac{133}{165}   & \frac{640}{2233}    & 0                   &  0 \\[\dis]
\frac{1}{111540}      & 0                 & 0                   & \frac{17209}{58047} & \frac{997}{5082} \\[\dis]
\frac{114721}{892320} & \frac{1}{146850}  & \frac{17}{506}      & \frac{385}{1759}    & \frac{2921}{203280} \\[\dis]
\frac{47}{1248}       & \frac{413}{5874}  & \frac{1}{102718}    & \frac{2915}{10554}  & \frac{4381}{203280} \\[\dis]
\frac{36}{169}        & \frac{22}{267}    & \frac{18674}{51359} & \frac{1}{116094}    & \frac{3252}{4235} \\[\dis]
\frac{276953}{446160} & \frac{1009}{24475}& \frac{16239}{51359} & \frac{1100}{5277}   & \frac{1}{101640}
\end{pmatrix} 
\in \Q_+^{6 \times 5}.
\]
The nonnegative rank of~$M$ over~$\R$ is~$5$.
The nonnegative rank of~$M$ over~$\Q$ is~$6$.
\end{theorem}

The rest of this paper is devoted to the proof of Theorem~\ref{thm:matrixWithDifferentRanks}. 

The matrix $M$ is stochastic.
The matrix~$M^{\prime}$ has a stochastic 5-dimensional NMF $M^{\prime} = W \cdot H^{\prime}$ with $W$, $H^{\prime}$ as follows:
\begin{align}
W
&= 
\begin{pmatrix}
0 & \frac{5}{7}+\frac{5\sqrt{2}}{77} & \frac{15+5\sqrt{2}}{77} & 0 & 0 \\[\dis]
0 & 0 & 0 & \frac{20 + 2 \sqrt{2}}{77} & \frac{48 - 8\sqrt{2}}{187} \\[\dis]
\frac{\sqrt{2}}{11} & 0 & \frac{4 - \sqrt{2}}{77} & \frac{3}{14} + \frac{\sqrt{2}}{308} & \frac{14 - 8\sqrt{2}}{187} \\[\dis]
\frac{-1 + \sqrt{2}}{11} & \frac{4 + \sqrt{2}}{77} & 0 & \frac{39}{154} + \frac{5\sqrt{2}}{308} & \frac{21 - 12 \sqrt{2}}{187} \\[\dis]
\frac{8 -4\sqrt{2}}{11} & \frac{12 -4\sqrt{2}}{77} & \frac{4}{11} & 0 & \frac{104 + 28\sqrt{2}}{187} \\[\dis]
\frac{4+2\sqrt{2}}{11} & \frac{6 -2\sqrt{2}}{77} & \frac{30 -4\sqrt{2}}{77} & \frac{3}{11} - \frac{\sqrt{2}}{22} & 0
\end{pmatrix}, \label{eq:leftIrrationalMatrix}
\end{align}
\begin{align*}
H^{\prime}
&=
\begin{pmatrix}
\frac{1+\sqrt{2}}{4} & 0 & \frac{\sqrt{2}}{11} & \frac{1}{4} - \frac{\sqrt{2}}{8} & 0 & \frac{1}{6} + \frac{\sqrt{2}}{12} \\[\dis]
0 & \frac{1}{2} - \frac{\sqrt{2}}{8} & 1 - \frac{\sqrt{2}}{11} & 0 & 0 & 0 \\[\dis]
\frac{3 - \sqrt{2}}{4} & \frac{1}{2} + \frac{\sqrt{2}}{8} & 0 & 0 & 0 & 0 \\[\dis]
0 & 0 & 0 & 0 & \frac{21}{34} + \frac{7 \sqrt{2}}{68} & \frac{5}{6} - \frac{\sqrt{2}}{12}\\[\dis]
0 & 0 & 0 & \frac{3}{4} + \frac{\sqrt{2}}{8} & \frac{13}{34} - \frac{7 \sqrt{2}}{68} & 0 
\end{pmatrix}.
\end{align*}
The matrix~$W_\eps$ has a stochastic 5-dimensional NMF $W_\eps = W \cdot H_\eps$ with $H_\eps$ as follows:
\[
\begin{pmatrix}
\frac{30419}{40560} + \frac{28679\sqrt2}{162240} \hspace{-1mm}& \frac{-2728}{46725} + \frac{5791\sqrt2}{140175}  & \frac{2741}{98049} - \frac{642\sqrt2}{32683} & \frac{-689}{10554} + \frac{15595\sqrt2}{337728} & \frac{389}{1848} - \frac{5501\sqrt2}{36960} \\[\dis]
0                                                \hspace{-1mm}& \frac{163318}{140175} - \frac{7277\sqrt2}{62300} & \frac{5958}{32683} - \frac{50543\sqrt2}{392196} & 0                                           & 0 \\[\dis]
0                                                \hspace{-1mm}& \frac{-2137}{20025} + \frac{6047\sqrt2}{80100}   & \frac{11062}{14007} + \frac{8321\sqrt2}{56028}  & 0                                           & 0 \\[\dis]
\frac{7443}{8840} - \frac{51313\sqrt2}{86190}    \hspace{-1mm}& 0                                                & 0                                    & \frac{148897}{179418} + \frac{172627\sqrt2}{1435344} & \frac{-1741}{26180} + \frac{1847\sqrt2}{39270} \\[\dis]
\frac{-408157}{689520} + \frac{1154473\sqrt2}{2758080} \hspace{-1mm}& 0                                          & 0                                    & \frac{7039}{29903} - \frac{318541\sqrt2}{1913792}    & \frac{134461}{157080} + \frac{1163\sqrt2}{11424}
\end{pmatrix}
\]
Hence, the matrix~$M$ has a stochastic 5-dimensional NMF as follows:
\begin{equation}
M = W \cdot \begin{pmatrix} H^{\prime} & H_\eps \end{pmatrix}. \label{eq:fact}
\end{equation}
We refer the reader to~\cite{maple-sheet}
for a Maple worksheet with calculations of the paper.
\begin{remark}
The columns of $M$ and~$W$ span the same vector space.
It follows that the \emph{restricted} nonnegative ranks of~$M$ over $\R$
and~$\Q$ are $5$ and~$6$, respectively.
In fact, the authors of this paper previously exhibited a rational nonnegative
matrix whose restricted nonnegative ranks over $\R$
and~$\Q$ differ~\cite{16CKMSW-ICALP}.
\end{remark}
We fix the matrices $M, M', W_\eps, W, H', H_\eps$ for the remainder of the paper.


\subsection{Types of Factorizations} \label{sub:profiles}

Let $M = L \cdot R$ be a stochastic NMF of inner dimension at most~5. (As
argued in Section~\ref{subsec-nonnegRank}, without loss of generality
we may consider only stochastic NMFs of $M$.)  Let us introduce the
following notation:
\begin{itemize}
\item $k$ is the number of columns in~$L$ whose first and second coordinates are~$0$,
\item $k_1$ is the number of columns in~$L$ whose first coordinate is strictly positive and second coordinate is~$0$, and
\item $k_2$ is the number of columns in~$L$ whose second coordinate is strictly positive and first coordinate is~$0$. 
\end{itemize}
Clearly, the factorization $M = L \cdot R$ corresponds to representing each column of $M$ as a convex
combination of the columns of $L$, with the coefficients of the convex
combination specified by the entries of~$R$. As $L$ has at most five columns, 
\begin{equation}
k+k_{1}+k_{2} \le 5. \label{ineq:profiles-totalsum}
\end{equation}
Since the columns $M_{:, 1}, M_{:, 2}, M_{:, 3}$  are linearly independent, the matrix $L$ has at least three columns whose second coordinate is~$0$. Likewise, since the columns $M_{:, 4}, M_{:, 5}, M_{:, 6}$ are linearly independent, $L$ has at least
three columns whose first coordinate is~$0$. That is,
\begin{equation}
k + k_1 \ge 3 \quad\text{ and }\quad k + k_2 \ge 3. \label{ineq:profiles-colsMindependent}
\end{equation}
Together with (\ref{ineq:profiles-totalsum}), this implies that $2 k \ge 6 - k_1 - k_2 \ge 1 + k$, and therefore $k \ge 1$.

The columns $M_{:, 1}, M_{:, 2}, M_{:, 3}$ have first coordinate strictly positive and second coordinate~$0$, while the columns $M_{:, 4}, M_{:, 5}, M_{:, 6}$ have second coordinate strictly positive and first coordinate~$0$. Therefore, in order for these columns to be covered by columns of $L$, we need to have: 
\begin{equation}
k_1 \ge 1 \quad\text{ and }\quad k_2 \ge 1. \label{ineq:profiles-atleast1}
\end{equation}
Together with (\ref{ineq:profiles-totalsum}), this implies that $k \le 5 - k_{1} - k_{2} \le 3$. 
We conclude that $k \in \{ 1, 2, 3 \}$. More precisely, it is now a consequence of inequalities \eqref{ineq:profiles-totalsum}, (\ref{ineq:profiles-colsMindependent}), and (\ref{ineq:profiles-atleast1}) that
the NMF $M = L \cdot R$ has (at least) one of the following four \emph{types}:
\begin{enumerate}
\item $k=1$, $k_1=2$, $k_2=2$;
\item $k=2$, $k_1=1$, $k_2 \in \{0,1,2\}$;
\item $k=2$, $k_2=1$, $k_1 \in \{0,1,2\}$;
\item $k=3$, $k_1=1$, $k_2=1$.
\end{enumerate}
These four types are illustrated in Figure~\ref{fig:zero-pattern} for
NMFs of inner dimension~5.

\begin{figure}
\begin{center}
\begin{tikzpicture}
\matrix [matrix of nodes,row sep=0,column 1/.style={nodes={minimum
width=1em}},left delimiter={(},right delimiter={)}, font=\scriptsize, label={above:type 1}]at (0,0)
{
0 & + & + & 0 & 0 \\
0 & 0 & 0 & + & + \\
$\cdot$ & $\cdot$  & $\cdot$  & $\cdot$  & $\cdot$  \\
$\cdot$ & $\cdot$  & $\cdot$  & $\cdot$  & $\cdot$  \\
$\cdot$ & $\cdot$  & $\cdot$  & $\cdot$  & $\cdot$ \\
$\cdot$ & $\cdot$  & $\cdot$  & $\cdot$  & $\cdot$ \\ 
};

\matrix [matrix of nodes,row sep=0,column 1/.style={nodes={minimum
width=1em}},left delimiter={(},right delimiter={)},  font=\scriptsize,label={above:type 2}] at (4.1,0)
{
0 & 0 & + & $\cdot$ & $\cdot$ \\
0 & 0 & 0 & + & + \\
$\cdot$ & $\cdot$  & $\cdot$  & $\cdot$  & $\cdot$  \\
$\cdot$ & $\cdot$  & $\cdot$  & $\cdot$  & $\cdot$  \\
$\cdot$ & $\cdot$  & $\cdot$  & $\cdot$  & $\cdot$ \\
$\cdot$ & $\cdot$  & $\cdot$  & $\cdot$  & $\cdot$ \\ 
};

\matrix [matrix of nodes,row sep=0,column 1/.style={nodes={minimum
width=1em}},left delimiter={(},right delimiter={)},  font=\scriptsize,label={above:type 3}]at (8.2,0)
{
0 & 0 & 0 & + & + \\
0 & 0 & + & $\cdot$ & $\cdot$ \\
$\cdot$ & $\cdot$  & $\cdot$  & $\cdot$  & $\cdot$  \\
$\cdot$ & $\cdot$  & $\cdot$  & $\cdot$  & $\cdot$  \\
$\cdot$ & $\cdot$  & $\cdot$  & $\cdot$  & $\cdot$ \\
$\cdot$ & $\cdot$  & $\cdot$  & $\cdot$  & $\cdot$ \\ 
};

\matrix [matrix of nodes,row sep=0,column 1/.style={nodes={minimum
width=1em}},left delimiter={(},right delimiter={)}, font=\scriptsize, label={above:type 4}] at (12.3,0)
{
+ & 0 & 0 & 0 & 0 \\
0 & + & 0 & 0 & 0 \\
$\cdot$ & $\cdot$  & $\cdot$  & $\cdot$  & $\cdot$  \\
$\cdot$ & $\cdot$  & $\cdot$  & $\cdot$  & $\cdot$  \\
$\cdot$ & $\cdot$  & $\cdot$  & $\cdot$  & $\cdot$ \\
$\cdot$ & $\cdot$  & $\cdot$  & $\cdot$  & $\cdot$ \\ 
};
\end{tikzpicture}
\end{center}
\caption{In any $5$-dimensional NMF $M = L \cdot R$, the matrix $L$ matches one of the four patterns above, up to a permutation of its columns. Here $+$ denotes any strictly positive number. }
\label{fig:zero-pattern}
\end{figure}

In Section~\ref{sec-proof} we prove the following proposition:

\begin{proposition} \label{prop:types}
Let $M$ be the matrix from Theorem~\ref{thm:matrixWithDifferentRanks} and
$W$ the matrix from Equation~(\ref{eq:leftIrrationalMatrix}). 
\begin{enumerate}[leftmargin=*]
\item If $M = L \cdot R$ is a type-1 NMF
then $W_{:,1}$ is a column of~$L$, and thus $L$ is not rational.
\item The matrix~$M$ has no type-2 NMF.
\item The matrix~$M$ has no type-3 NMF.
\item The matrix~$M$ has no type-4 NMF.
\end{enumerate}
\end{proposition}

Using this proposition we can prove Theorem~\ref{thm:matrixWithDifferentRanks}:
\begin{proof}[Proof of Theorem~\ref{thm:matrixWithDifferentRanks}]
Due to the NMF of~$M$ stated in (\ref{eq:fact}), the nonnegative rank of~$M$ is at most~$5$.
If there existed an at most 4-dimensional NMF of~$M$ then, as $k+k_{1}+k_{2} \le 4$, it would have to have type 2 or~3, but those types are excluded by items (2)~and~(3) of Proposition~\ref{prop:types}. Hence the nonnegative rank of~$M$ over~$\R$ equals~$5$.

Since $M = I \cdot M$ (where $I$ denotes the $6 \times 6$ identity matrix), the nonnegative rank of~$M$ over~$\Q$ is at most~$6$.
By Proposition~\ref{prop:types} there is no 5-dimensional NMF $M = L \cdot R$ with $L$ rational.
Hence, the nonnegative rank of~$M$ over~$\Q$ equals~$6$.
\end{proof}

\section{Proof of Proposition~\ref{prop:types}} \label{sec-proof}

It remains to prove Proposition~\ref{prop:types}.  To rule out type-4
NMFs we use constraint propagation in order to prove that the
inequalities required for type-4 NMFs are contradictory, see
Section~\ref{subsec-Profile3}.  To rule out rational NMFs of
types 1, 2, and 3, we employ geometric arguments concerning nested
polygons in the plane (see
Sections~\ref{subsec-Profile1}--\ref{subsec-Profile2b}).  These
arguments rely on a geometric interpretation of the specific NMF
$M = W \cdot \begin{pmatrix} H^{\prime} & H_\eps \end{pmatrix}$
given by Equation~\eqref{eq:fact}.  More precisely, we define a polytope
$\mathcal{P}\subseteq\mathbb{R}^3$, shown in
Figure~\ref{fig-3d-parallel-eyed}, such that each of the columns of
$M$ and $W$ can be associated with a point in~$\mathcal{P}$. The
points associated with the columns of~$M$ lie in the convex hull of
those associated with the columns of $W$
(cf.~\cite{gillis2012geometric}).

\subsection{Geometry behind the Proof of Proposition~\ref{prop:types}}\label{sec-proof-geometry}

To set up this geometric connection, observe that the matrix~$M$ is
stochastic and has rank~$4$, and hence the columns of~$M$ affinely
span a 3-dimensional affine subspace $\mathcal{V} \subseteq \R^6$.
All vectors in~$\mathcal{V}$ are pseudo-stochastic.
The set of stochastic vectors in~$\mathcal{V}$
(equivalently, the nonnegative vectors in $\mathcal{V}$) form a
3-dimensional polytope, say \mbox{$\mathcal{P}'\subseteq \mathcal{V}$}. 
Clearly we have $M_{:,i} \in \mathcal{P}'$ for each
$i\in\{1,\ldots,11\}$.

\subsubsection*{Parameterization} \label{sub:parametrization}

We will now fix a particular parameterization of~$\mathcal{V}$
and~$\mathcal{P}'$; that is, we define an injective affine function
$f: \R^3 \to \R^6$ and a polytope $\mathcal{P}\subseteq\mathbb{R}^3$
such that $f(\R^3) = \mathcal{V}$ and $f(\mathcal{P}) = \mathcal{P}'$.
Let $f: \R^3 \to \R^6$ be the function with $f(x) = C x + d$ for each $x \in \R^3$, where:
\begin{align*}
C = 
\frac{1}{11} \cdot \begin{pmatrix}
0  & 10 & 0   \\
0  &  0 & 4   \\
-1 & -2 & 1/2 \\
-1 &  0 & 5/2 \\
4  &  0 & 0   \\
-2 & -8 & -7  
\end{pmatrix} \in \Q^{6 \times 3} 
~~~\text{ and }~~
&
d = 
\frac{1}{11} \cdot \begin{pmatrix}
0 \\ 
0 \\ 
2 \\ 
1 \\ 
0 \\ 
8
\end{pmatrix} \in \Q^{6 \times 1}.
\end{align*}
Note that the map $f$ is injective.

Defining
\[
r_1 = \begin{pmatrix}3/4\\ 1/8\\ 0\end{pmatrix}\negthickspace,\ %
r_2 = \begin{pmatrix}3/4\\ 1/2\\ 0\end{pmatrix}\negthickspace,\ %
r_3 = \begin{pmatrix}3/11\\ 17/22\\ 0\end{pmatrix}\negthickspace,\ %
r_4 = \begin{pmatrix}2\\ 0\\ 1/2\end{pmatrix}\negthickspace,\ %
r_5 = \begin{pmatrix}1/2\\ 0\\ 3/4\end{pmatrix}\negthickspace,\ %
r_6 = \begin{pmatrix}1/6\\ 0\\ 7/12\end{pmatrix}\negthickspace,\ %
\]
we have $f(r_{i}) = M^{\prime}_{:, i}
$ for each $i \in \{1, 2, \ldots, 6\}$, and defining
\[
q_1^\eps = \begin{pmatrix}\frac{99}{169}    \\[\dis] 0        \\[\dis] \frac{1}{40560}    \end{pmatrix}\negthickspace,\ %
q_2^\eps = \begin{pmatrix}\frac{121}{534}   \\[\dis] \frac{133}{150} \\[\dis] 0           \end{pmatrix}\negthickspace,\ %
q_3^\eps = \begin{pmatrix}\frac{9337}{9338} \\[\dis] \frac{64}{203}  \\[\dis] 0           \end{pmatrix}\negthickspace,\ %
q_4^\eps = \begin{pmatrix}\frac{1}{42216}   \\[\dis] 0        \\[\dis] \frac{17209}{21108}\end{pmatrix}\negthickspace,\ %
q_5^\eps = \begin{pmatrix}\frac{813}{385}   \\[\dis] 0        \\[\dis] \frac{997}{1848}   \end{pmatrix}\negthickspace,\ %
\]
we have $f(q_i^\eps) = (W_\eps)_{:, i}$ for each $i \in \{1, 2, \ldots, 5\}$.
Thus, all columns of~$M$ lie in the image of~$f$.
It follows that $f(\R^3) = \mathcal{V}$.

Let $\mathcal{P}$ be the $3$-dimensional polytope defined by $\{\,x \in \R^3 \mid f(x) \ge 0\,\}$.
Then $f(\mathcal{P}) = \mathcal{P}'$.
Furthermore, $r_i \in \mathcal{P}$, as $f(r_i) = M'_{:,i} \in \mathcal{P}'$ for all $i \in \{1, 2, \ldots, 6\}$.
Likewise we 
have $q_i^\eps \in \mathcal{P}$, as $f(q_i^\eps) = (W_\eps)_{:, i} \in \mathcal{P}'$ for all $i \in \{1, 2, \ldots, 5\}$.

\begin{figure}
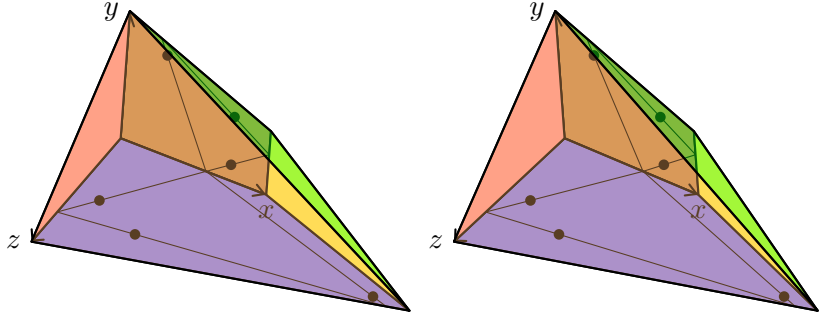

\begin{center}
\objectThreeD{32}{40}{-5}{2.2}\hspace{-13mm}
\objectThreeD{32}{40}{5}{2.2}
\end{center}
\caption{
The two images show orthogonal projections of the 3-dimensional  polytope~$\mathcal{P}$. The black dots indicate 6 interior points: 3 points ($r_1$, $r_2$, $r_3$) on the brown $xy$-face, and 3 points ($r_4$, $r_5$, $r_6$) on the blue $xz$-face.
(The images form a stereo pair
 intended for ``parallel-eye'' watching:
 to see the polytope in~3D,
 look at the left and right projections
 with your left and right eyes respectively
 at the same time,
 as described, e.g., in~\cite{DSimanekParallelEye}.)
}
\label{fig-3d-parallel-eyed}
\end{figure} 
Figure~\ref{fig-3d-parallel-eyed} visualizes~$\mathcal{P}$, which has 6 faces
corresponding to the inequalities of the system $C x + d \ge 0$.
In more detail, $\mathcal{P}$ is the intersection of the following half-spaces:
$y \ge 0$ (blue), 
$z \ge 0$ (brown),
$-\frac{1}{2} x - y + \frac{1}{4} z + 1 \ge 0$ (green),
$-x + \frac{5}{2} z + 1 \ge 0$ (yellow),
$x \ge 0$ (pink),
$-\frac{1}{4} x - y - \frac{7}{8} z + 1 \ge 0$ (transparent front).
The figure also shows the position of the points $r_1, \ldots, r_6$ (black dots).%
\footnote{%
    In~\cite{16CKMSW-ICALP} the authors of this article used the same
    polytope~$\mathcal{P}$ and the same points $r_1, \ldots, r_6$ (see
    Figure~\ref{fig-3d-parallel-eyed}) to prove a related result about the
    \emph{restricted} nonnegative rank.%
}

In fact, the columns of~$W$ are also in~$\mathcal{P}' \subseteq \mathcal{V}$.
Indeed, defining
\[
q_1^* = 
\begin{pmatrix}
2 - \sqrt{2} \\[\dis] 0 \\[\dis] 0
\end{pmatrix}\negthickspace,\ %
q_2^* = 
\begin{pmatrix}
\frac{3 - \sqrt{2}}{7} \\[\dis] \frac{11 + \sqrt{2}}{14} \\[\dis] 0
\end{pmatrix}\negthickspace,\ %
q_3^* = 
\begin{pmatrix} 
1 \\[\dis] \frac{3 + \sqrt{2}}{14} \\[\dis] 0
\end{pmatrix}\negthickspace,\ %
q_4^* = 
\begin{pmatrix}
0 \\[\dis] 0 \\[\dis] \frac{10 + \sqrt{2}}{14}
\end{pmatrix}\negthickspace,\ %
q_5^* = 
\begin{pmatrix}
\frac{26 + 7 \sqrt{2}}{17} \\[\dis] 0 \\[\dis] \frac{12 - 2 \sqrt{2}}{17}
\end{pmatrix}\negthickspace,\ %
\]
we have $f(q_i^*) = W_{:,i} \in \mathcal{P}'$, hence $q_i^* \in \mathcal{P}$ for each $i \in \{1, 2, \ldots, 5\}$.
That is, in our NMF $M = W \cdot \begin{pmatrix} H^{\prime} & H_\eps \end{pmatrix}$, the columns of $M$ and the columns of~$W$ span the same vector space.
Such NMFs are called \emph{restricted} in~\cite{gillis2012geometric} and~\cite{16CKMSW-ICALP}.
Applying the inverse of the map~$f$ column-wise to the NMFs $M' = W \cdot H'$ and $W_\eps = W \cdot H_\eps$, we obtain
\begin{equation} \label{eq:inverseEqOfF}
\begin{aligned}
\begin{pmatrix}
r_1 & r_2 & r_3 & r_4 & r_5 & r_6
\end{pmatrix}
&=
\begin{pmatrix}
q_1^* & q_2^* & q_3^* & q_4^* & q_5^*
\end{pmatrix}
\cdot H'
\qquad \text{and} \\
\begin{pmatrix}
q_1^\eps & q_2^\eps & q_3^\eps & q_4^\eps & q_5^\eps
\end{pmatrix}
&=
\begin{pmatrix}
q_1^* & q_2^* & q_3^* & q_4^* & q_5^*
\end{pmatrix}
\cdot H_\eps\,,
\end{aligned}
\end{equation}
respectively.
Recall that the matrix $H'$ is stochastic;
hence \eqref{eq:inverseEqOfF} implies that the points $r_i$ and~$q_i^\eps$ are contained in the convex hull of the points~$q_i^*$.
In Figure~\ref{fig-3d-parallel-eyed}, points
$q_1^*, q_2^*, q_3^*$ are the vertices of the triangle on the brown
$x y$-face, while points $q_1^*, q_4^*, q_5^*$ are the vertices of the
triangle on the blue $x z$-face. The former triangle contains
$r_1, r_2, r_3$, while the latter triangle contains $r_4, r_5,
r_6$. Points $q_1^\eps, \ldots, q_5^\eps$ (not shown in
Figure~\ref{fig-3d-parallel-eyed}) are close to
$q_1^*, \ldots, q_5^*$, with $q_2^\eps, q_3^\eps$ lying in the
interior of the triangle on the $x y$-face and
$q_1^\eps, q_4^\eps, q_5^\eps$ lying in the interior of the triangle
on the $x z$-face.

It is important to note that when we exclude certain NMFs $M = L \cdot R$ in Sections \ref{subsec-Profile1}--\ref{subsec-Profile3}, we cannot a priori assume that the columns of~$L$
are in~$\mathcal{V}$.

\subsubsection*{Nested Polygons}
In this subsection, we focus on the two faces of polytope $\mathcal{P}$ that contain the interior points $r_1, r_2 , r_3$ and $r_4, r_5, r_6$, respectively called $\mathcal{P}_0$ and $\mathcal{P}_1$.

Let us write $\mathcal{V}_0\subseteq\mathbb{R}^6$ for the affine span of
$M_{:, 1}, M_{:, 2}, M_{:, 3}$.  We can also characterize $\mathcal{V}_0$ as the image
of the $xy$-plane in $\mathbb{R}^3$ under the map
$f:\mathbb{R}^3\rightarrow\mathbb{R}^6$.  Indeed, we have $f(r_1)=M_{:, 1}$,
$f(r_2)=M_{:, 2}$, and $f(r_3)=M_{:, 3}$.  Thus the image of the $xy$-plane
under $f$ is a two-dimensional affine space that includes
$\mathcal{V}_0$ and hence is equal to $\mathcal{V}_0$.  Define the
polygon $\mathcal{P}_0 \subseteq \mathbb{R}^3$ by
$\mathcal{P}_0 = \{ (x,y,0)^\top : (x,y,0)^\top \in \mathcal{P} \}$.  Then $f$
restricts to a bijection between $\mathcal{P}_0$ and the set of
nonnegative vectors in $\mathcal{V}_0$.  
We have the following lemma:

\begin{figure*}[t]
\begin{center}
\begin{tikzpicture}[
scale=7,
dot/.style={circle,fill=black,minimum size=4pt,inner sep=0pt,outer sep=-1pt}
]
\draw[->, >=stealth',>=angle 60] (0,0) -- (1.05,0) node[right] {$x$};
\draw[->, >=stealth',>=angle 60] (0,0) -- (0,1.05) node[above] {$y$};
\foreach \x in {0.25,0.5,0.75,1.0}
    \draw[shift={(\x,0)}] (0pt,0pt) -- (0pt,-0.4pt) node[below] {$\x$};
\foreach \x in {0.25,0.5,0.75,1.0}
    \draw[shift={(0,\x)}] (0,0) -- (-0.4pt,0) node[left] {$\x$};
\node at (-0.04,-0.04) {$0$};
    
\draw[thick,fill=brown!95!orange!50] (0,0) coordinate (O1) -- (1,0) coordinate (O2) -- (1,1/2) coordinate (O3) -- (0,1) coordinate (O4) -- (O1);
\node[label={[label distance=-0.5mm,xshift=0mm]0:$(1,\frac12,0)^\top$}] at (O3) {};
\coordinate (I1) at (3/4,1/8);
\coordinate (I2) at (3/4,1/2);
\coordinate (I3) at (3/11,17/22);
\node[dot,label={[label distance=0mm]90:$r_1$}] at (I1) {};
\node[dot,label={[label distance=1mm]-90:$r_2$}] at (I2) {};
\node[dot,label={[label distance=0mm]180:$r_3$}] at (I3) {};
\coordinate (M1) at ({2-sqrt(2)}, 0);
\coordinate (M2) at  (intersection cs:
       first line={(M1)--(I1)},
       second line={(O2)--(O3)});
\coordinate (M3) at  (intersection cs:
       first line={(M2)--(I2)},
       second line={(O3)--(O4)});
\node[dot,label={[label distance=-0.5mm, xshift=0mm]-90:$q_1^*$}] at (M1) {};
\node[dot,label={[label distance=0mm]0:$q_3^*$}] at (M2) {};
\node[dot,label={90:$q_2^*$}] at (M3) {};
\draw (M1) -- (M2) -- (M3) -- (M1);

\coordinate (M1+) at (0.59, 0);
\coordinate (M2+) at  (intersection cs:
       first line={(M1+)--(I1)},
       second line={(O2)--(O3)});
\coordinate (M3+) at  (intersection cs:
       first line={(M2+)--(I2)},
       second line={(O3)--(O4)});
\coordinate (M4+) at  (intersection cs:
       first line={(M3+)--(I3)},
       second line={(O1)--(O2)});

\coordinate (M1-) at (0.58, 0);
\coordinate (M2-) at  (intersection cs:
       first line={(M1-)--(I1)},
       second line={(O2)--(O3)});
\coordinate (M3-) at  (intersection cs:
       first line={(M2-)--(I2)},
       second line={(O3)--(O4)});
\coordinate (M4-) at  (intersection cs:
       first line={(M3-)--(I3)},
       second line={(O1)--(O2)});

\coordinate (q1) at (1/8,0);
\coordinate (q2) at (intersection cs:
       first line={(O1)--(I1)},
       second line={(O2)--(O3)});
\coordinate (q3) at (intersection cs:
       first line={(q2)--(I2)},
       second line={(O3)--(O4)});
\coordinate (q4) at (intersection cs:
       first line={(q3)--(I3)},
       second line={(O1)--(O4)});

\draw[dashed, red] (q1) -- (q2) -- (q3) -- (q4) -- (q1);
\node[dot,label={[label distance=0mm]270:$q_1$}] at (q1) {};
\node[dot,label={[label distance=0mm]0:$q_3$}] at (q2) {};
\node[dot,label={[label distance=0mm]90:$q_2$}] at (q3) {};
\node[dot,label={[label distance=0mm]180:$q_4$}] at (q4) {};


\coordinate (M1spy) at (0.602,0.027);
\end{tikzpicture}
\end{center}
\caption{The outer polygon is $\mathcal{P}_0$ (after identifying the
  $xy$-plane in $\R^3$ with $\R^2$).
  The triangle with solid boundary is the supporting polygon~$\mathcal{S}_{q_1^*}$, where $q_1^* = (2-\sqrt{2},0,0)^\top$.
  The  quadrilateral with dashed boundary is the supporting polygon $\mathcal{S}_{q_1}$ for $q_1=(\frac{1}{8},0,0)^\top$.  }
\label{fig-concave-display}
\end{figure*}

\begin{lemma}
  Let $\mathcal{R}_0 \subseteq \mathcal{P}_0$ be the polygon with
  vertices $r_1, r_2, r_3$ (see Figure~\ref{fig-concave-display}).
  Write $q_1=(u,0,0)^\top$, where $0 \leq u \leq 1$.  If the
  supporting polygon $\mathcal{S}_{q_1}$ nested between $\mathcal{R}_0$
  and $\mathcal{P}_0$ has three vertices, then $u \geq 2-\sqrt{2}$.
\label{lem:ineq}
\end{lemma}
\begin{proof}
  Assume that $\mathcal{S}_{q_1}$ has three vertices and $0 \le u \le 2 - \sqrt{2}$.
  It suffices to show that these assumptions imply $u = 2 - \sqrt{2}$.
  Moving anti-clockwise, let the
  vertices of $\mathcal{S}_{q_1}$ be $q_1$, $q_3$, and $q_2$.
  It follows by
  elementary geometry that: (i)~the line segment $q_1 q_3$ passes
  through~$r_1$ and $q_3$ lies on the right edge of $\mathcal{P}_0$,
  and (ii)~the line segment $q_3 q_2$ passes through~$r_2$ and $q_2$
  lies on the upper edge of $\mathcal{P}_0$.
  Figure~\ref{fig-concave-display} shows the situations
  $u=2-\sqrt{2}$ and $u=\frac{1}{8}$.
  
Writing $q_3=(1,\frac{v}{2},0)^\top$ and
  $q_2=(1-w,\frac{1}{2}+\frac{w}{2},0)^\top$, where $0\leq v,w \leq 1$, the
  collinearity conditions (i) and (ii) entail (see
  Section~\ref{sec:NPP}):
\begin{alignat}{2}
 \begin{vmatrix}
 u & 0 & 1\\[\dis]
 1 & \textstyle\frac{v}{2} & 1\\[\dis]
 \textstyle\frac{3}{4} & \frac{1}{8} & 1 
\end{vmatrix} 
&= \frac12 u v - \frac18 u - \frac38 v + \frac18 &&= 0 \qquad \text{and} \label{eq-colinear1} \\
\begin{vmatrix}
 1 & \textstyle\frac{v}{2} & 1\\[\dis]
 1-w & \textstyle\frac{1}{2}+\frac{w}{2} & 1\\[\dis]
 \textstyle\frac{3}{4} & \frac{1}{2} & 1 
\end{vmatrix} 
&= \frac12 v w - \frac18 v - \frac38 w + \frac18 &&=
0\,. \label{eq-colinear2} \intertext{%
  The assumption that $\mathcal{S}_{q_1}$ is the triangle $\triangle q_1 q_3 q_2$ entails that vertices $q_2,q_1,r_3$ are in
  anti-clockwise order.  This implies:}
\begin{vmatrix}
  1-w & \textstyle\frac{1}{2}+\frac{w}{2} & 1\\[\dis]
 u & 0 & 1 \\[\dis]
 \textstyle\frac{3}{11} & \frac{17}{22} & 1 
\end{vmatrix}
&= -\frac12 w u + \frac{10}{11} w + \frac{3}{11} u - \frac{7}{11} &&\geq 0. \label{eq-colinear3}
\intertext{%
We use \eqref{eq-colinear1} and~\eqref{eq-colinear2} to eliminate variables $v,w$ from the inequality~\eqref{eq-colinear3}, obtaining:
}
&\frac{15}{22 (8 u - 5)} \cdot (u^2 - 4 u + 2) && \ge 0. \nonumber
\end{alignat}
The only solution with $0\leq u \le 2-\sqrt{2}$ is $u = 2-\sqrt{2}$.
\end{proof}

Let us write $\mathcal{V}_1\subseteq\mathbb{R}^6$ for the affine span
of $M_{:, 4}, M_{:, 5}, M_{:, 6}$. 
We can also characterize
$\mathcal{V}_1$ as the image of the $xz$-plane in $\mathbb{R}^3$ under
the map $f:\mathbb{R}^3\rightarrow\mathbb{R}^6$.  Indeed, we have $f(r_4)=M_{:, 4}$,
$f(r_5)=M_{:, 5}$, and $f(r_6)=M_{:, 6}$.  Thus the image of the
$xz$-plane under $f$ is a two-dimensional affine space that includes
$\mathcal{V}_1$ and hence is equal to $\mathcal{V}_1$.  Define
the polygon $\mathcal{P}_1 \subseteq \mathbb{R}^3$ by
$\mathcal{P}_1 = \{ (x,0,z)^\top : (x,0,z)^\top \in \mathcal{P} \}$.  Then $f$
restricts to a bijection between $\mathcal{P}_1$ and the set of
nonnegative vectors in $\mathcal{V}_1$.
We have the following lemma:

\begin{figure*}[t]
\begin{center}
\begin{tikzpicture}[
scale=5,
dot/.style={circle,fill=black,minimum size=4pt,inner sep=0pt,outer sep=-1pt}]
\draw[->,>=stealth',>=angle 60] (0,0) -- (2.4,0) node[right] {$x$};
\draw[->,>=stealth',>=angle 60] (0,0) -- (0,1.35) node[above] {$z$};
\foreach \x in {0.25,0.5,0.75,1.0,1.25,1.5,1.75,2.0,2.25}
    \draw[shift={(\x,0)}] (0pt,0pt) -- (0pt,-0.4pt) node[below] {$\x$};

\foreach \x in {0.25,0.5,0.75,1.0,1.25}
    \draw[shift={(0,\x)}] (0,0) -- (-0.4pt,0) node[left] {$\x$};
\node at (-0.06,-0.06) {$0$};
\draw[thick, fill=blue!60!orange!35] (0,0) coordinate (O1) -- (1,0) coordinate (O2) -- (9/4,1/2) coordinate (O5) -- (0,8/7) coordinate (O6) -- (O1);
\node[label={[label distance=-0.5mm,xshift=0mm]0:$(\frac94,0,\frac12)^\top$}] at (O5) {};
\node[label={[xshift=-1mm,yshift=2mm]0:$(0,0,\frac87)^\top$}] at (O6) {};

\coordinate (I4) at (2,1/2);
\coordinate (I5) at (1/2,3/4);
\coordinate (I6) at (1/6,7/12);
\node[dot,label={[label distance=0mm]-90:$r_4$}] at (I4) {};
\node[dot,label={[label distance=0mm]90:$r_5$}] at (I5) {};
\node[dot,label={[label distance=0mm]0:$r_6$}] at (I6) {};
\coordinate (M1) at ({2-sqrt(2)}, 0);
\coordinate (M4) at  (intersection cs:
       first line={(M1)--(I4)},
       second line={(O5)--(O6)});
\coordinate (M5) at  (intersection cs:
       first line={(M4)--(I5)},
       second line={(O1)--(O6)});
\draw (M1) -- (M4) -- (M5) -- (M1);

\coordinate (M1+) at (7/8, 0);
\coordinate (M4+) at  (intersection cs:
       first line={(M1+)--(I4)},
       second line={(O5)--(O6)});
\coordinate (M5+) at  (intersection cs:
       first line={(M4+)--(I5)},
       second line={(O1)--(O6)});
\coordinate (M6+) at  (intersection cs:
       first line={(M5+)--(I6)},
       second line={(O1)--(O2)});

\draw[dashed, red] (M1+) -- (M6+) -- (M5+) -- (M4+) -- (M1+);

\coordinate (M1-) at (0., 0);
\coordinate (M4-) at  (intersection cs:
       first line={(M1-)--(I4)},
       second line={(O5)--(O6)});
\coordinate (M5-) at  (intersection cs:
       first line={(M4-)--(I5)},
       second line={(O1)--(O6)});
\coordinate (M6-) at  (intersection cs:
       first line={(M5-)--(I6)},
       second line={(O1)--(O2)});
       
\node[dot,label=90:$q_1^*$] at (M1) {};
\node[dot,label=90:$q_1$] at (M1+) {};
\node[dot,label={[label distance=0mm]90:$q_5^*$}] at (M4) {};
\node[dot,label={[label distance=0mm,yshift=1mm]180:$q_4^*$}] at (M5) {};


\end{tikzpicture}
\end{center}
\caption{The outer polygon is $\mathcal{P}_1$ (after identifying the
  $xz$-plane in $\mathbb{R}^3$ with $\mathbb{R}^2$).
  The triangle with solid
  boundary is the supporting polygon $\mathcal{S}_{q_1^*}$, where
  $q_1^* = (2-\sqrt{2},0,0)^\top$.
  The quadrilateral with dashed boundary is the supporting polygon $\mathcal{S}_{q_1}$ for
  $q_1=(\frac{7}{8},0,0)^\top$.}
\label{xz-plane}
\end{figure*}
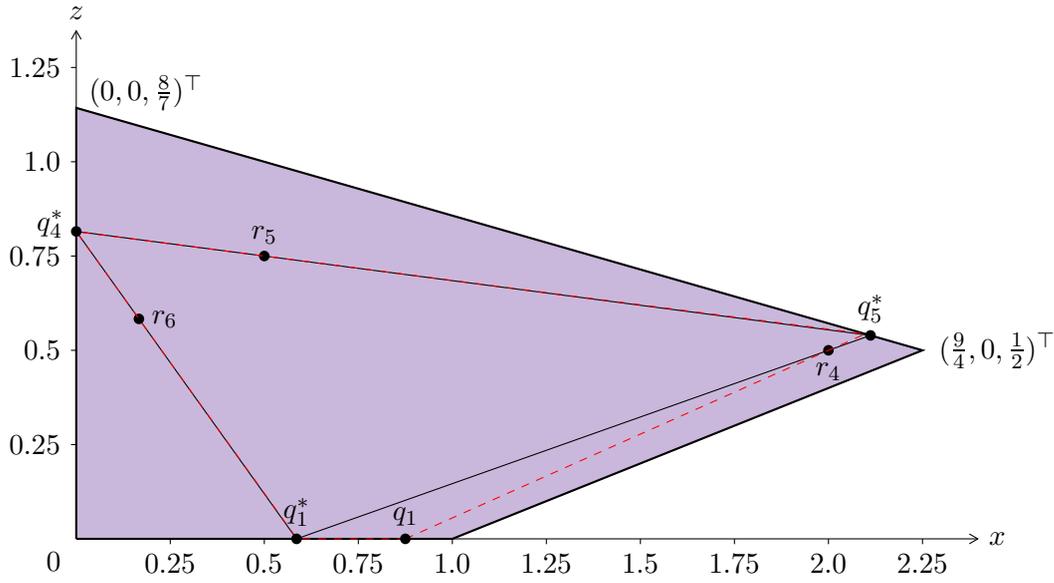  

\begin{lemma}
  Let $\mathcal{R}_1 \subseteq \mathcal{P}_1$ be the polygon with
  vertices $r_4, r_5, r_6$ (see Figure~\ref{xz-plane}).  Write
  $q_1=(u,0,0)^\top$, where $0\leq u \leq 1$.  If the supporting
  polygon $\mathcal{S}_{q_1}$ nested between $\mathcal{R}_1$ and
  $\mathcal{P}_1$ has three vertices, then $u \leq 2-\sqrt{2}$.
\label{lem:ineq2}
\end{lemma}
\begin{proof}
  Assume that $\mathcal{S}_{q_1}$ has three vertices and $2 - \sqrt{2} \le u \le 1$.
  It suffices to show that these assumptions imply $u = 2 - \sqrt{2}$.
  Moving anti-clockwise, let the vertices
  of $\mathcal{S}_{q_1}$ be $q_1$, $q_5$, and $q_4$.  
  It follows by elementary
  geometry that: (i)~the line segment $q_1 q_5$ passes through $r_4$
  and $q_5$ lies on the upper edge of $\mathcal{P}_1$, and (ii)~the
  line segment $q_5 q_4$ passes through $r_5$ and $q_4$ lies on the
  left edge of $\mathcal{P}_1$.
  Figure~\ref{xz-plane} shows the situations $u=2-\sqrt{2}$ and $u=\frac78$.
  
Writing $q_5=(\frac{9 - 9 v}{4},0,\frac{7 + 9 v}{14})^\top$ and
  $q_4=(0,0,\frac{8-8w}{7})^\top$, where $0\leq v,w \leq 1$, the
  collinearity conditions (i) and (ii) entail (see Section~\ref{sec:NPP}):
\begin{alignat}{2}
 \begin{vmatrix}
 u & 0 & 1\\[\dis]
 \textstyle\frac{9 - 9 v}{4} & \frac{7 + 9 v}{14} & 1\\[\dis]
 2 & \textstyle\frac{1}{2} & 1 
\end{vmatrix} 
&= \frac{9}{14} u v - \frac{135}{56} v + \frac18 &&= 0 \qquad \text{and} \label{eq-colinear4} \\
\begin{vmatrix}
\textstyle\frac{9 - 9 v}{4} & \frac{7 + 9 v}{14} & 1\\[\dis]
 0 & \frac{8-8w}{7} & 1 \\[\dis]
 \textstyle\frac{1}{2} & \frac{3}{4} & 1 
\end{vmatrix} 
&= \frac{18}{7} v w - \frac{9}{16} v - 2 w + \frac{9}{16} &&=
0 \label{eq-colinear5} \intertext{%
  The assumption that $\mathcal{S}_{q_1}$ is the triangle $\triangle q_1 q_5 q_4$
  entails that vertices $q_4,q_1,r_6$ are in
  anti-clockwise order.  This implies:}
\begin{vmatrix}
0 & \frac{8-8w}{7} & 1 \\[\dis]
 u & 0 & 1 \\[\dis]
 \textstyle\frac{1}{6} & \frac{7}{12} & 1 
\end{vmatrix} 
&= \frac87 w u - \frac{4}{21} w - \frac{47}{84} u + \frac{4}{21} &&\geq 0 \label{eq-colinear6}
\intertext{%
We use \eqref{eq-colinear4}~and~\eqref{eq-colinear5} to eliminate variables
$v,w$ from the inequality~\eqref{eq-colinear6}, obtaining:
}
&\frac{-10}{21 (2 u - 7)} \cdot (u^2 - 4 u + 2) &&\ge 0 \nonumber
\end{alignat}
The only solution with $2-\sqrt{2} \le u \le 1$ is $u = 2-\sqrt{2}$.
\end{proof}

\subsection{Type 1}\label{subsec-Profile1}

In this section we prove Proposition~\ref{prop:types}~(1),
implying that any type-1 NMF of~$M$ requires irrational numbers
(our argument will, in fact, only depend on the matrix~$M'$ and not
on $W_\eps$).
Consider a type-1 NMF $M = L\cdot R$, i.e., such that $k=1$ and $k_1=k_2=2$.
After a suitable permutation of its columns, $L$ matches the pattern
\[ L=
\begin{pmatrix}
0 & + & + & 0 & 0 \\
0 & 0 & 0 & + & + \\
\cdot & \cdot  & \cdot  & \cdot  & \cdot  \\
\cdot & \cdot  & \cdot  & \cdot  & \cdot  \\
\cdot & \cdot  & \cdot  & \cdot  & \cdot \\
\cdot & \cdot  & \cdot  & \cdot  & \cdot  
\end{pmatrix},
\]
where $+$ denotes any strictly positive number.
It follows from the zero pattern of~$M$ that $M_{:,1}, M_{:, 2}, M_{:, 3}$ all lie in the convex hull
of $L_{:, 1}, L_{:, 2}, L_{:, 3}$, and
$M_{:, 4}, M_{:, 5}, M_{:, 6}$ all lie in the convex hull of
$L_{:, 1}, L_{:, 4}, L_{:, 5}$.
Equivalently, there exist stochastic matrices $R_0, R_1 \in \R_+^{3 \times 3}$ such that
\begin{align}
\begin{pmatrix}
M_{:, 1} & M_{:, 2} & M_{:, 3}
\end{pmatrix}
&=
\begin{pmatrix}
L_{:, 1} & L_{:, 2} & L_{:, 3}
\end{pmatrix}
\cdot
R_0
\qquad\text{and}\qquad \label{eq:type1-0} \\
\begin{pmatrix}
M_{:, 4} & M_{:, 5} & M_{:, 6}
\end{pmatrix}
&=
\begin{pmatrix}
L_{:, 1} & L_{:, 4} & L_{:, 5}
\end{pmatrix}
\cdot
R_1\,. \nonumber
\end{align}

Consider the polygon $\mathcal{P}_0 \subseteq \R^3$ and the affine space $\mathcal{V}_0 \subseteq \R^6$ from
Section~\ref{sec-proof-geometry}.
The affine span of $L_{:, 1}, L_{:, 2}, L_{:, 3}$
includes $\mathcal{V}_0$ and has dimension at most two, and hence is
equal to $\mathcal{V}_0$.  In particular,
$L_{:, 1}, L_{:, 2}, L_{:, 3}$ must all lie in $\mathcal{V}_0$.  Since
$L_{:, 1}, L_{:, 2}, L_{:, 3}$ are moreover nonnegative, there are
uniquely defined points $q_1,q_2,q_3 \in \mathcal{P}_0$ such that
$f(q_i) = L_{:, i}$ for $i \in \{1,2,3\}$.
Applying the inverse of the map~$f$ column-wise to~\eqref{eq:type1-0}, we obtain
\[
\begin{pmatrix}
r_1 & r_2 & r_3
\end{pmatrix}
=
\begin{pmatrix}
q_1 & q_2 & q_3
\end{pmatrix}
\cdot
R_0\,,
\]
so the convex hull of $q_1,q_2,q_3$ includes $r_1,r_2,r_3$.
In other words, triangle $\triangle q_1 q_2 q_3$ is nested between $\triangle r_1 r_2 r_3$ and polygon~$\mathcal{P}_0$.
Since $L_{:, 1}$ has $0$ in its first two coordinates, by inspecting the definition of the map~$f$ we see that $q_1 = (u, 0, 0)^\top$ for some $u \in \R$.
By Lemma~\ref{lem:NPP} it follows that the supporting polygon $\mathcal{S}_{q_1}$ has three vertices.
Hence Lemma~\ref{lem:ineq} implies $u \ge 2 - \sqrt{2}$.

Considering the polygon~$\mathcal{P}_1$ from Section~\ref{sec-proof-geometry}, we have $q_1 \in \mathcal{P}_1$ (recall that $f(q_1) = L_{:,1}$).
Arguing as in the case of $\mathcal{P}_0$, there are uniquely defined points $q_4,q_5 \in \mathcal{P}_1$ such that $f(q_i) = L_{:, i}$ for $i \in \{4,5\}$.
Similarly as before, triangle $\triangle q_1 q_4 q_5$ is nested between $\triangle r_4 r_5 r_6$ and $\mathcal{P}_1$.
Then Lemmas~\ref{lem:NPP} and~\ref{lem:ineq2} imply $u \le 2 - \sqrt{2}$, thus $q_1 = (2 - \sqrt{2}, 0, 0)^\top = q_1^*$.
Hence 
\[
L_{:,1} = f(q_1) = f(q_1^*) = W_{:,1}.
\]
Proposition~\ref{prop:types}~(1) follows.

We remark that this argument can be strengthened 
to show that any type-1 NMF
of~$M$ coincides with the one given in~\eqref{eq:fact}, up to a permutation
of the columns of~$W$ and the rows of~$\begin{pmatrix} H^{\prime} & H_\eps \end{pmatrix}$;
see Appendix~\ref{app:uniqueness}.

\subsection{Type 2}\label{subsec-Profile2}
In this section we exclude type-2 NMFs, i.e., we prove Proposition~\ref{prop:types}~(2). 
Towards a contradiction, suppose there is a stochastic and at most 5-dimensional NMF $M = L \cdot R$ with $k=2$ and $k_1=1$.
Without loss of generality, the first three columns of $L$ match the following pattern:
\begin{equation*}
\kbordermatrix{ 
& L_{:, 1} & L_{:, 2}& L_{:, 3} \\
& 0 & 0 & +\\
& 0 & 0 & 0\\
& \cdot & \cdot  & \cdot\\
& \cdot & \cdot  & \cdot\\
& \cdot & \cdot  & \cdot\\
& \cdot & \cdot  & \cdot  
}
\end{equation*}
and the remaining columns have a strictly positive second coordinate.
It follows from the zero pattern of~$M$ that $M_{:,1}, M_{:, 2}, M_{:, 3}$ all lie in the convex hull of $L_{:, 1}, L_{:, 2}, L_{:, 3}$.

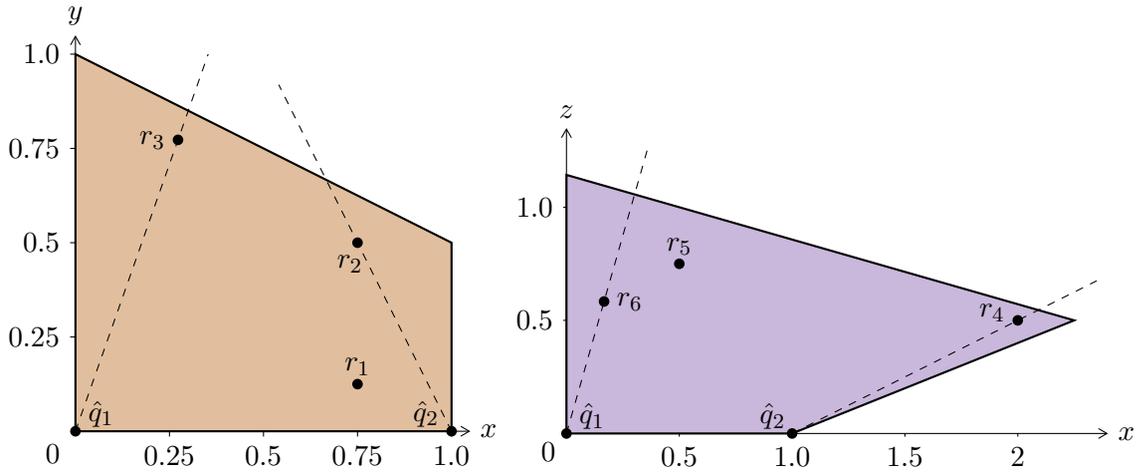
\begin{figure*}[t]
\begin{center}
\begin{tikzpicture}[
scale=5,
dot/.style={circle,fill=black,minimum size=4pt,inner sep=0pt,outer sep=-1pt}
]
\draw[->, >=stealth',>=angle 60] (0,0) -- (1.05,0) node[right] {$x$};
\draw[->, >=stealth',>=angle 60] (0,0) -- (0,1.05) node[above] {$y$};
\foreach \x in {0.25,0.5,0.75,1.0}
    \draw[shift={(\x,0)}] (0pt,0pt) -- (0pt,-0.4pt) node[below] {$\x$};

\foreach \x in {0.25,0.5,0.75,1.0}
    \draw[shift={(0,\x)}] (0,0) -- (-0.4pt,0) node[left] {$\x$};
\draw[thick,fill=brown!95!orange!50] (0,0) coordinate (O1) -- (1,0) coordinate (O2) -- (1,1/2) coordinate (O3) -- (0,1) coordinate (O4) -- (O1);

\node at (-0.06,-0.06) {$0$};

\coordinate (I1) at (3/4,1/8);
\coordinate (I2) at (3/4,1/2);
\coordinate (I3) at (3/11,17/22);
\node[dot,label={[label distance=-0.5mm]90:$r_1$}] at (I1) {};
\node[dot,label={[label distance=0mm,xshift=-1mm]-90:$r_2$}] at (I2) {};
\node[dot,label={[label distance=0mm]180:$r_3$}] at (I3) {};
\coordinate (M1) at (0, 0);
\coordinate (M2) at (1, 0);
\coordinate (M3) at  (intersection cs:
       first line={(M1)--(I3)},
       second line={(0,1)--(1,1)});
\coordinate (M4) at (intersection cs:
       first line={(M2)--(I2)},
       second line={(0,1)--(1,0.85)});

\node[dot,label={[label distance=0mm,yshift=+2.3mm]0:$\hat{q}_1$}] at (M1) {};
\node[dot,label={[label distance=0mm,yshift=+2.3mm]180:$\hat{q}_2$}] at (M2) {};

\draw[dashed] (M1) -- (M3);
\draw[dashed] (M2) -- (M4);

\end{tikzpicture}
\begin{tikzpicture}[
scale=3,
dot/.style={circle,fill=black,minimum size=4pt,inner sep=0pt,outer sep=-1pt}]
\draw[->,>=stealth',>=angle 60] (0,0) -- (2.4,0) node[right] {$x$};
\draw[->,>=stealth',>=angle 60] (0,0) -- (0,1.35) node[above] {$z$};
\foreach \x in {0.5,1.0,1.5,2}
    \draw[shift={(\x,0)}] (0pt,0pt) -- (0pt,-0.4pt) node[below] {$\x$};

\foreach \x in {0.5,1.0}
    \draw[shift={(0,\x)}] (0,0) -- (-0.4pt,0) node[left] {$\x$};
\draw[thick, fill=blue!60!orange!35] (0,0) coordinate (O1) -- (1,0) coordinate (O2) -- (9/4,1/2) coordinate (O5) -- (0,8/7) coordinate (O6) -- (O1);

\node at (-0.08,-0.08) {$0$};

\coordinate (I4) at (2,1/2);
\coordinate (I5) at (1/2,3/4);
\coordinate (I6) at (1/6,7/12);
\node[dot,label={[label distance=0mm,yshift=1mm]180:$r_4$}] at (I4) {};
\node[dot,label={[label distance=-0.5mm]90:$r_5$}] at (I5) {};
\node[dot,label={[label distance=0mm]0:$r_6$}] at (I6) {};

\coordinate (M1) at (0, 0);
\coordinate (M2) at (1, 0);
\coordinate (M3) at  (intersection cs:
       first line={(M1)--(I6)},
       second line={(0,1.25)--(1,1.25)});
\coordinate (M4) at (intersection cs:
       first line={(M2)--(I4)},
       second line={(2.35,0)--(2.35,1)});
\node[dot,label={[label distance=0mm,yshift=2.3mm]0:$\hat{q}_1$}] at (M1) {};
\node[dot,label={[label distance=0mm,yshift=2.3mm,xshift=0.8mm]180:$\hat{q}_2$}] at (M2) {};
\draw[dashed] (M1) -- (M3);
\draw[dashed] (M2) -- (M4);
\end{tikzpicture}
\end{center}
\caption{
Left: there is no point
  $q_3$ in the quadrilateral~$\mathcal{P}_0$ such that $\triangle \hat{q}_1 \hat {q}_2 q_3$
  includes both $r_2$ and $r_3$.
Right: there is no point
  $q_3$ in the quadrilateral~$\mathcal{P}_1$ such that $\triangle \hat{q}_1 \hat {q}_2 q_3$
  includes both $r_4$ and $r_6$.}
\label{fig-type23}
\end{figure*}  

Consider again the polygon $\mathcal{P}_0 \subseteq \R^3$.
For the purposes of the following
argument, $\mathcal{P}_0$ is visualized on the left of Figure~\ref{fig-type23}.
Reasoning analogously as in Section~\ref{subsec-Profile1}, there are
unique points $q_1, q_2, q_3 \in \mathcal{P}_0$ with
$f(q_i) = L_{:, i}$ for $i \in \{1, 2, 3\}$, and the
convex hull of $q_1,q_2,q_3$ includes $r_1,r_2,r_3$.  Since $L_{:, 1}$
and~$L_{:,2}$ have $0$ in their first two rows, inspecting the
definition of the map~$f$, we see that $q_1$ and~$q_2$ lie on the
$x$-axis in $\mathbb{R}^3$.  Thus, writing $\hat{q}_1 = (0,0,0)^\top$
and $\hat{q}_2 = (1,0,0)^\top$, triangle~$\triangle \hat{q}_1 \hat {q}_2 q_3$
includes triangle~$\triangle q_1q_2q_3$, and hence also contains the points
$r_1,r_2,r_3$.  But clearly there is no point $q_3\in\mathcal{P}_0$
such that $\triangle \hat{q}_1 \hat {q}_2 q_3$ includes both $r_2$ and
$r_3$ (see, e.g., Figure~\ref{fig-type23}, left), which is a contradiction.  Thus we have proved
Proposition~\ref{prop:types}~(2).

\subsection{Type 3}\label{subsec-Profile2b}
In this section we exclude type-3 NMFs, i.e., we prove Proposition~\ref{prop:types}~(3). 
The reasoning is entirely analogous to Section~\ref{subsec-Profile2}.
Towards a contradiction, suppose there is a stochastic and at most 5-dimensional NMF $M = L \cdot R$ with $k=2$ and $k_2=1$.
Consider again the polygon $\mathcal{P}_1 \subseteq \R^3$. 
For the purposes of the following argument, $\mathcal{P}_1$ is visualized on the right of Figure~\ref{fig-type23}.
Analogously to Section~\ref{subsec-Profile2}, there are points $q_1, q_2, q_3 \in \mathcal{P}_1$ whose convex hull includes $r_4,r_5,r_6$, and $q_1$ and~$q_2$ lie on the $x$-axis in $\mathbb{R}^3$.
Thus, writing $\hat{q}_1 = (0,0,0)^\top$
and $\hat{q}_2 = (1,0,0)^\top$, triangle~$\triangle \hat{q}_1 \hat {q}_2 q_3$
includes the points $r_4,r_5,r_6$.  But clearly there is no point $q_3\in\mathcal{P}_1$
such that $\triangle \hat{q}_1 \hat {q}_2 q_3$ includes both $r_4$ and
$r_6$ (see, e.g., Figure~\ref{fig-type23}, right), which is a contradiction.  Thus we have proved
Proposition~\ref{prop:types}~(3).

\subsection{Type 4}\label{subsec-Profile3}

\def\Mahsa#1{{\sf \fbox{$\clubsuit$ Mahsa: #1 $\clubsuit$}}}
\newcommand{\tW}{\widetilde{W}}
\newcommand{\tL}{L}
\newcommand{\tR}{R}
\newcommand{\rat}{\mathbb{Q}}
\newcommand{\real}{\mathbb{R}}
\newcommand{\oups}{\frac{1}{10^5}}
\newcommand{\ignore}[1]{}
\newcommand{\tWep}{W_{\epsilon}}
\newcommand{\magenta}[1]{{\color{magenta}#1}}

In this section we exclude type-4 NMFs, i.e., we prove Proposition~\ref{prop:types}~(4). 
In fact, Sections \ref{subsec-Profile1}--\ref{subsec-Profile2b} prove the stronger
result that there is no rational NMF of types~1, 2, 3
for the matrix~$M'$ alone.
Here we spell out the role of~$W_\eps$, effectively
explaining why the matrix~$M = \begin{pmatrix} M' & W_\eps \end{pmatrix}$
is defined the way it is.

Observe that adding to~$M'$ new columns from
the convex hull of the columns of~$W$
shrinks the set of possible
nonnegative factorizations. 
Given this, our goal is to find a matrix satisfying
the following desiderata:
\begin{itemize}
\item its entries are rational,
\label{desiderata}
\item its columns belong to the convex hull of the columns of $W$, and
\item it has no type-4 NMF.
\end{itemize}
The first two items ensure
that $M$, while being rational, admits a nonnegative factorization
with left factor~$W$, ensuring
that the nonnegative rank of~$M$ over~\R is (at most)~$5$.
The third condition,
combined with the arguments from Sections \ref{subsec-Profile1}--\ref{subsec-Profile2b}, ensures that the nonnegative rank
of~$M$ over~\Q is~$6$.

Whilst the matrix~$W$ manifestly fails the first desideratum,
it satisfies the last two:
\begin{claim}
\label{c:no-4-noeps}
The matrix $W$ and, therefore, the matrix~%
$\overline M = \begin{pmatrix} M^{\prime} & W \end{pmatrix}$
have no type-4 NMF.
\end{claim}
This reasoning motivates
the main
technical result of this section,
a strengthening of Claim~\ref{c:no-4-noeps}
showing that no matrix in a suitably small neighborhood
of~$W$ admits a type-4 NMF:

\newcommand{\stmtlemtypeFourConstraints}{
For all stochastic matrices~$\tW \in \R_{+}^{6 \times 5}$  satisfying the entry-wise constraints given in Figure~\ref{fig:constraints}, there exists no type-4 NMF
$\tW = \tL \cdot \tR$.
}

\begin{lemma}\label{l:constraints}
\stmtlemtypeFourConstraints
\end{lemma}

In particular, the constraints of Lemma~\ref{l:constraints}, and in fact all three desiderata,
are satisfied by the matrix~$W_\eps$
from Theorem~\ref{thm:matrixWithDifferentRanks};
see Figure~\ref{fig:Weps-rounded}.
Therefore, the matrix~$M = \begin{pmatrix} M^{\prime} & W_\eps \end{pmatrix}$
has no type-4 NMF,
thus concluding the proof of Proposition~\ref{prop:types}~(4).

\begin{remark}
The existence of a suitable matrix~$W_\eps$
can be
understood in terms of upper semi-continuity of the nonnegative 
rank over the reals~\cite{BocciCR11}
and
can be alternatively demonstrated
using a non-constructive argument that assumes
only Claim~\ref{c:no-4-noeps} instead of the (stronger)
Lemma~\ref{l:constraints};
see Appendix~\ref{app:topological} for details.
We are, however, not aware of
a simpler proof of Claim~\ref{c:no-4-noeps}.
\end{remark}



%
\begin{figure*}[t]
\begin{equation*}
\arraycolsep=9pt
\tW = 
\kbordermatrix{
					 & \tW_{:,1}						& \tW_{:,2} 	& \tW_{:,3}											& \tW_{:,4} 		& \tW_{:,5} 					\\[4pt]
\tW_{1,:}  & 0				 						& 0.8 \le \cdot 	  &  0.286 \le \, \cdot \, \le 0.287  & 0       			&	 0 									\\[4pt]
\tW_{2,:}  & \cdot \le \epsilon  				& 0         	& 0      						  					&   0.29	\le \cdot  	&  0.196 \le \cdot  \\[4pt]
\tW_{3,:}  &            					& \cdot \le \epsilon&    0.0335 \le \cdot      						&       0.21\le \cdot & \cdot  \le 0.015  \\[4pt]
\tW_{4,:}  &          						& 0.07 \le \cdot		& \cdot  \le \epsilon								  &      0.27\le \cdot  & \cdot  \le 0.022 \\[4pt]
\tW_{5,:}  &            					&        		  & 				         							& \cdot \le  \epsilon  & 0.767 \le \cdot\\[4pt]
\tW_{6,:}  & 0.62 \le \cdot    					&          	  & \cdot   \le 0.32								    & \cdot \le 0.21   	  & \cdot \le  \epsilon
}
\hspace*{1.5em} 
\end{equation*}
\caption{Entry-wise constraints, where $\epsilon=10^{-5}$.}
\label{fig:constraints}
\end{figure*}
\begin{figure*}[t]
\begin{center}
\begin{tikzpicture}
\matrix [matrix of nodes,column sep = 10, left delimiter={(},right delimiter={)}, label={left:$W_\eps \approx$~~~~}]
{
0 & 0.81 & 0.2866 & 0 & 0 \\
$0.9 \cdot 10^{-5}$ & 0 & 0 & 0.296 & 0.1962 \\
0.1 & $0.7 \cdot 10^{-5}$ & 0.03360 & 0.219 & 0.0144 \\
0.04 & 0.0703 & $0.97 \cdot 10^{-5}$ & 0.276 & 0.0216 \\
0.2 & 0.08 & 0.4 & $0.9 \cdot 10^{-5}$ & 0.7679 \\
0.621 & 0.04 & 0.316 & 0.208 & $0.98 \cdot 10^{-5}$\\
};
\end{tikzpicture}
\end{center}
\caption{Matrix $W_\eps$ with entries rounded off.}
\label{fig:Weps-rounded}
\end{figure*}

\subsubsection*{Proof of Lemma~\ref{l:constraints}}

The idea of the proof is to
derive a contradiction from the assumption that there exists a
stochastic matrix \mbox{$\tW \in \R_{+}^{6 \times 5}$} that satisfies
the constraints in Figure~\ref{fig:constraints} and has an
NMF $\tW = \tL \cdot \tR$ of type~4, i.e., such that  
$L$ matches the following zero pattern:
\[ L=
\begin{pmatrix}
+ & 0 & 0 & 0 & 0 \\
0 & + & 0 & 0 & 0 \\
\cdot & \cdot  & \cdot  & \cdot  & \cdot  \\
\cdot & \cdot  & \cdot  & \cdot  & \cdot  \\
\cdot & \cdot  & \cdot  & \cdot  & \cdot \\
\cdot & \cdot  & \cdot  & \cdot  & \cdot  
\end{pmatrix}.
\]
  To this end, we use constraint propagation to successively derive lower and upper bounds for various entries of the matrices~$\tL$ and~$\tR$ until we reach a
contradiction.

In our proof of the lemma, we use the following two
assumptions, which are with no loss of generality:
\begin{enumerate}
\renewcommand{\theenumi}{\textup{(A\arabic{enumi})}}
\renewcommand{\labelenumi}{\theenumi}
\item
\label{a:1}
$\tL_{6,3} = \max \{ \tL_{6,3}, \tL_{6,4}, \tL_{6,5} \}$ and
\item
\label{a:2}
$\tL_{5,4} = \max \{ \tL_{5,4}, \tL_{5,5} \}$.
\end{enumerate}
Technically, these assumptions are only used
in the proof of Claim~\ref{c:y1+y2} below.

We first demonstrate that Lemma~\ref{l:constraints}
follows from the following two claims:

\begin{claim}
\label{c:y1+y2}
$\tL_{6,3} \ge 0.61$,
$\tL_{5,4} \ge 0.9539$.
\end{claim}

\begin{claim}
\label{c:o1+o2}
$\tL_{4,3} \ge 0.346$,
$\max\{\tL_{3,3}, \tL_{3,4}\} \ge 0.0465$.
\end{claim}

\begin{proof}[Proof of Lemma~\ref{l:constraints}]
Take the second inequality of Claim~\ref{c:o1+o2}
and consider two cases: 
\begin{itemize}[leftmargin=*]
\item
First suppose that in Claim~\ref{c:o1+o2} it holds that
$\max\{\tL_{3,3},\tL_{3,4}\} = \tL_{3,3}$.
Then $\tL_{3,3} \ge 0.0465$.
Further, Claims~\ref{c:o1+o2} and~\ref{c:y1+y2} give lower bounds on $\tL_{4,3}$ and $\tL_{6,3}$, respectively.
Since the elements of each column of $\tL$ sum up to~$1$, it follows
$0.0465 + 0.346 + 0.61 \le \tL_{3,3} + \tL_{4,3} + \tL_{6,3}\leq 1$.
This is a contradiction.
\item
Otherwise, $\max\{\tL_{3,3},\tL_{3,4}\} = \tL_{3,4} \ge 0.0465$.
Recall that Claim~\ref{c:y1+y2} gives $\tL_{5,4} \geq 0.9539$.
Hence $0.0465 + 0.9539 \le \tL_{3,4} + \tL_{5,4} \le 1$, which is also a contradiction.
\end{itemize}
\end{proof}

Our two goals now are to prove Claims~\ref{c:y1+y2} and~\ref{c:o1+o2}.
We achieve this using a sequence of auxiliary statements.

\begin{claim}
\label{c:14}
$0.29\leq \tR_{2,4}$,
$0.196\leq\tR_{2,5}$.
\end{claim}

\begin{proof}
Observe that all columns of~$\tW$  lie in the convex hull of columns of~$\tL$
since $\tW_{:,j}=  \tL \cdot \tR_{:,j}$.
Consider the $4$th and $5$th columns~$\tW_{:,4},\tW_{:,5}$.
Since~$\tL_{:,2}$ is the only column of~$\tL$ with strictly positive
second component, we have $\tW_{2,4}= \tL_{2,2} \cdot \tR_{2,4}$
and $\tW_{2,5} = \tL_{2,2} \cdot \tR_{2,5}$.
Therefore,
\begin{equation*}
\begin{aligned}
&0.29\leq \tW_{2,4}= \tL_{2,2} \cdot \tR_{2,4} \leq \tR_{2,4},
&\quad \quad \quad&0.196\leq\tW_{2,5}=\tL_{2,2} \cdot \tR_{2,5} \leq \tR_{2,5},
\end{aligned}
\end{equation*}
implying the claim.
\end{proof}

By omitting nonnegative terms
from the equality $\tW_{i,j}=\tL_{i,:} \cdot \tR_{:,j}$,
we obtain the inequality $\tL_{i,k} \cdot \tR_{k,j} \leq  \tW_{i,j}$,
which holds for all~$1\leq k\leq 5$.
We can thus compute an upper bound on~$\tL_{i,k}$ (resp., on~$\tR_{k,j}$) if we 
know a lower bound on~$\tR_{k,j}$ (resp., on~$\tL_{i,k}$). 
We refer to this as computing \emph{simple upper bounds} through $\tW_{i,j}$.

\begin{claim}
\label{c:15}
The matrix $\tL$ satisfies the following constraints:
\begin{equation*}
\kbordermatrix{~ &\tL_{:,1}  &\tL_{:,2} &\tL_{:,3} &\tL_{:,4} &\tL_{:,5}\cr
  \tL_{1,:} &&&&&\\
  \tL_{2,:} &&0.8 \leq \cdot &&&\\
        \tL_{3,:} &&\cdot \leq 0.077    &&&\\
	\tL_{4,:} &&\cdot \leq 0.12     &&&\\
	\tL_{5,:} &&\cdot \leq 4\epsilon&&&\\
	\tL_{6,:} &&\cdot \leq 6 \epsilon&&&
	}
\end{equation*}
\end{claim}

\begin{proof}
First note that, by Claim~\ref{c:14}, $0.196 \leq \tR_{2,5}$.
This lets us derive
the following simple upper bounds through~$\tW_{3,5}, \tW_{4,5}$ and $\tW_{6,5}$:
\begin{equation*}
\begin{aligned}
\tL_{3,2} & \leq \frac{0.015}{0.196} \le 0.077, &\quad \quad&
\tL_{4,2} & \leq \frac{0.022}{0.196} \le 0.12,  &\quad \quad&
\tL_{6,2} & \leq \frac{\epsilon}{0.196} \le 6 \epsilon.
\end{aligned}
\end{equation*}
%
The remaining bounds are obtained as follows.
The lower bound on~$\tR_{2,4}$, taken from Claim~\ref{c:14},
gives a simple upper bound~$\tL_{5,2} \leq \epsilon / 0.29 \leq 4 \epsilon$
through~$\tW_{5,4}$.
Now the upper bounds on~$\tL_{i,2}$, where $3 \leq i \leq 6$,
result in the inequality~%
$\tL_{2,2} \geq 1 - (0.077 + 0.12 +10  \epsilon) \geq 0.8$.
\end{proof}

Now we are able to prove Claim~\ref{c:y1+y2}.

\begin{proof}[Proof of Claim~\ref{c:y1+y2}]
We use the bounds of Claim~\ref{c:15} and
the inequalities $\tW_{2,1} \le \epsilon$,
                 $\tW_{6,1} \ge 0.62$, and
                 $\tW_{5,5} \ge 0.767$
from the statement of Lemma~\ref{l:constraints}.

To begin with,
the first column of~$\tW$ lies in the convex hull of $\tL_{:,2}$ and~$\tL_{:,3},\tL_{:,4},\tL_{:,5}$.
From the lower bound~$\tL_{2,2}\geq 0.8$ we compute the simple upper bound $\tR_{2,1}\leq 2 \epsilon$
through~$\tW_{2,1}$. 
By our assumption~\ref{a:1},
$\tL_{:,3}$ has the largest $6$th coordinate among $\tL_{:,3},\tL_{:,4}$ and~$\tL_{:,5}$, so from
\[
0.62\leq \tW_{6,1}= \tL_{6,:} \cdot \tR_{:,1} \leq \tL_{6,2}\cdot \tR_{2,1}+ (1-\tR_{2,1}) \cdot \tL_{6,3}\leq \tL_{6,2}\cdot \tR_{2,1}+\tL_{6,3}\]
we obtain $\tL_{6,3} \geq  0.62- 6 \epsilon\cdot 2\epsilon \geq 0.61$, as claimed.

Furthermore,
the $5$th column of~$\tW$ also lies in the convex hull of $\tL_{:,2}$ and~$\tL_{:,3},\tL_{:,4},\tL_{:,5}$.
Recall that $\tL_{5,2}$ is at most~$4\epsilon$,  and $\tR_{2,5}$ is at least~$0.196$ by Claim~\ref{c:14}.
We then have
\[
0.767 \leq \tW_{5,5}=\tL_{5,:} \cdot \tR_{:,5} \leq \tL_{5,2}+ (1-\tR_{2,5}) \cdot \max\{\tL_{5,3},\tL_{5,4},\tL_{5,5}\}\,,\]
yielding the bound  $\max\{\tL_{5,3},\tL_{5,4},\tL_{5,5}\} \geq \frac{0.767- 4\epsilon}{1-0.196} \geq 0.9539$.
As we already know that $\tL_{6,3} \geq 0.61$, 
we deduce that
the maximum in the left-hand side cannot be attained by $\tL_{5,3}$,
since the column vector $\tL_{:,3}$ is stochastic.
Now, by our assumption~\ref{a:2}
we must have $\tL_{5,4} \geq 0.9539$.
\end{proof}

Our next goal is prove Claim~\ref{c:o1+o2}.
The following Claims~\ref{c:g+b} and~\ref{c:20}
take two consecutive steps in this direction.

\begin{claim}
\label{c:g+b}
$\tR_{5,2} \le 50\epsilon$,
$\tR_{5,3} \le 23\epsilon$.
\end{claim}

\begin{proof}
First note that
the matrix $\tL$ satisfies the following constraints:
\begin{equation}
\label{eq:16}
\kbordermatrix{~ &\tL_{:,1}  &\tL_{:,2} &\tL_{:,3} &\tL_{:,4} &\tL_{:,5}\cr
  \tL_{1,:} &&&&&\\
  \tL_{2,:} &&&&&\\
  \tL_{3,:} &&\cdot \le 0.077& &\cdot \le 0.0461&\\
	\tL_{4,:} &&\cdot \le 0.12&&\cdot \le 0.0461&\\
	\tL_{5,:} &&&&0.9539 \le \cdot &\\
	\tL_{6,:} &&& 0.61\le \cdot & &
	} \;.   
\end{equation}
Indeed,
the upper bounds on $\tL_{3,2}$ and~$\tL_{4,2}$ are taken verbatim
from Claim~\ref{c:15},
and the lower bounds on~$\tL_{6,3}$ and~$\tL_{5,4}$
from Claim~\ref{c:y1+y2}.
The latter bound implies the upper bounds on $\tL_{3,4}$ and~$\tL_{4,4}$.

We first prove the inequality~%
$\tR_{5,2} \le 50\epsilon$.
By multiplying the row vector $\begin{pmatrix}0 & 0 & 2 & 0 & 0 & -1\end{pmatrix}$ with~$\tW_{:,4} = \tL \cdot \tR_{:,4}$ we obtain $2\tW_{3,4}-\tW_{6,4} = (2\tL_{3,:}-\tL_{6,:})\cdot \tR_{:,4}$.
Since the $4$th column of~$\tW$ lies in the convex hull of $\tL_{:,2}$ and~$\tL_{:,3},\tL_{:,4},\tL_{:,5}$,
we also have
\begin{align*}
2\tW_{3,4}-\tW_{6,4} 
&\le\max \{2\tL_{3,2}-\tL_{6,2}, \ 2 \tL_{3,3}-\tL_{6,3}, \ 2 \tL_{3,4}-\tL_{6,4}, \ 2\tL_{3,5}-\tL_{6,5}\} \\
&\le\max \{2\tL_{3,2}-\tL_{6,2}, \ 2 \tL_{3,3}-\tL_{6,3}, \ 2 \tL_{3,4}-\tL_{6,4}\} + 2\tL_{3,5}.
\end{align*}
On the one hand, we have $2\tW_{3,4}-\tW_{6,4}\geq 0.21$,
because $\tW_{3,4} \geq 0.21$ and $\tW_{6,4}\le 0.21$. 
On the other hand, 
$2\tL_{3,3}-\tL_{6,3} \le 2\, (1 - \tL_{6,3}) - \tL_{6,3} = 2-3\tL_{6,3}$.
Hence,
\begin{align*}
0.21 
&\leq \max\{\underbrace{2\tL_{3,2}}_{{}\le 2\cdot 0.077}, \, \underbrace{2 -3\tL_{6,3}}_{{}\le 2 -3\cdot 0.61}, 
\,\underbrace{2 \tL_{3,4}}_{{}\le 2\cdot 0.0461}\} + 2\tL_{3,5}\;,
\end{align*}
where the inequalities are taken from~\eqref{eq:16}
and from the calculation above.
Therefore,
$
\tL_{3,5} \geq 0.02
$,
from which we derive the simple upper bound
$\tR_{5,2} \leq 50\epsilon$ through~$\tW_{3,2}$.

The second inequality,
$\tR_{5,3} \le 23\epsilon$,
is proved in a similar way.
By multiplying the row vector $\begin{pmatrix} 0 & 0 & 0 & 2 & 0 & -1\end{pmatrix}$
with~$\tW_{:,4} = \tL \cdot \tR_{:,4}$ we obtain $2\tW_{4,4}-\tW_{6,4} =(2\tL_{4,:}-\tL_{6,:})\cdot \tR_{:,4}$ and thus
\[
2\tW_{4,4}-\tW_{6,4} \le \max\{2\tL_{4,2}-\tL_{6,2}, \ 2 \tL_{4,3}-\tL_{6,3}, \ 2 \tL_{4,4}-\tL_{6,4}\} + 2\tL_{4,5}.
\]
On the one hand, we have $2\tW_{4,4}-\tW_{6,4}\geq 2\cdot 0.27  - 0.21 \geq  0.33$.
On the other hand, $2\tL_{4,3}-\tL_{6,3} \le 2  (1 - \tL_{6,3}) - \tL_{6,3} = 2-3\tL_{6,3}$.
Hence,
\begin{align*}
0.33 
&\leq \max\{\underbrace{2\tL_{4,2}}_{{}\le 2\cdot 0.12}, \, \underbrace{2 -3\tL_{6,3}}_{{}\le 2 -3\cdot 0.61}, 
\,\underbrace{2 \tL_{4,4}}_{{}\le 2\cdot 0.0461}\} + 2\tL_{4,5}\;,
\end{align*}
where the inequalities are again taken from~\eqref{eq:16}
and from the calculation above.
It follows that $\tL_{4,5}\geq 0.045$,
and we derive the simple upper bound
$\tR_{5,3} \leq \epsilon / 0.045 \leq 23 \epsilon$ through~$\tW_{4,3}$.
This completes the proof.
\end{proof}

\begin{claim}
\label{c:20}
The matrices $\tL$ and $\tR$ satisfy the following constraints:
\begin{equation*}
\kbordermatrix{~ &\tL_{:,1}  &\tL_{:,2} &\tL_{:,3} &\tL_{:,4} &\tL_{:,5}\cr
  \tL_{1,:} &&\hspace{-2mm}&\hspace{-2mm}&\hspace{-2mm}\\
  \tL_{2,:} &&\hspace{-2mm}&\hspace{-2mm}&\hspace{-2mm}\\
  \tL_{3,:} &\cdot \le 2\epsilon&\hspace{-2mm}&\hspace{-2mm}&\hspace{-2mm}\\
	\tL_{4,:} &\cdot \le 4\epsilon&\hspace{-2mm}&&\cdot \le 10\epsilon\hspace{-2mm}&\hspace{-2mm}\\
	\tL_{5,:} &&\hspace{-2mm}&&&\hspace{-2mm}\\
	\tL_{6,:} &&\hspace{-2mm}&& \hspace{-2mm}&\hspace{-2mm}
	} \;   
\kbordermatrix{~ &{\tR_{:,1}}&\tR_{:,2} &\tR_{:,3} &\tR_{:,4} &\tR_{:,5}\cr
  \tR_{1,:} && 0.8\le \cdot &0.286\le \cdot &\\
  \tR_{2,:} &&&&\\
  \tR_{3,:} &&&&\\
	\tR_{4,:} &&&&\\
	\tR_{5,:} &&\cdot \le 50\epsilon&\cdot \le23\epsilon& 
	}
\end{equation*}
\end{claim}

\begin{proof}
First note that the constraints
$\tR_{5,2} \le 50\epsilon$ and
$\tR_{5,3} \le 23\epsilon$
are already known to us
from Claim~\ref{c:g+b}.
We now show how to obtain the remaining five constraints.

Observe that
the column~$\tL_{:,1}$ is the only column of~$L$ that has a positive
first component; hence it is the only column of~$L$ that contributes to
the positive first component in the~$2$nd and~$3$rd columns of~$\tW$.
Therefore, the following inequalities indeed hold:
\begin{align*}
0.286  \leq \tW_{1,3}&=\tL_{1,1} \cdot  \tR_{1,3}\leq \tR_{1,3} 
    \quad \text{ and }\\
0.8 \leq  \tW_{1,2}&=\tL_{1,1} \cdot  \tR_{1,2} \le \tR_{1,2}.
\end{align*}
The latter inequality leads to the claimed simple upper bound $\tL_{3,1} \leq \epsilon / 0.8 \leq 2 \epsilon$ through~$\tW_{3,2}$.
We further derive the following simple upper bounds:
%
%
%
\begin{itemize}[leftmargin=*]
	\item $\tR_{1,3} \leq 0.287 / 0.8 \leq 0.36$ through~$\tW_{1,3}$,
            since $\tL_{1,1} \ge \tL_{1,1} \cdot \tR_{1,2} = \tW_{1,2} \geq 0.8$ by the above;
	\item $\tR_{3,3} \leq 0.32 / 0.61 \leq 0.53$ through~$\tW_{6,3}$,
            since $\tL_{6,3} \geq 0.61$ by Claim~\ref{c:y1+y2}.
\end{itemize}
Since~$\tW_{:,3}$ lies in the convex hull of $\tL_{:,1}$ and~$\tL_{:,3},\tL_{:,4},\tL_{:,5}$,
we can deduce that~%
$\tR_{4,3} = 1 - \tR_{1,3} - \tR_{3,3} - \tR_{5,3} \geq
 1- 0.36-0.53-23 \epsilon\geq 0.1$.
Using this lower bound, $\tR_{4,3} \ge 0.1$,
and the lower bound $\tR_{1,3} \ge 0.286$ obtained above,
we deduce,
through~$\tW_{4,3}$,
simple upper bounds
$\tL_{4,4} \leq 10\epsilon$ and $\tL_{4,1} \leq 4\epsilon$.
This concludes the proof.
\end{proof}

We are now ready to prove Claim~\ref{c:o1+o2}.

\begin{proof}[Proof of Claim~\ref{c:o1+o2}]
Here we will only use the result of the previous Claim~\ref{c:20}.

First note that
the $2$nd column of~$\tW$ lies in the convex hull of $\tL_{:,1}$, $\tL_{:,3}$,
$\tL_{:,4}$, and $\tL_{:,5}$.
We have
\begin{equation*}
\begin{aligned}
0.07 \leq \tW_{4,2}=\tL_{4,:} \cdot \tR_{:,2} &= \tL_{4,1} \tR_{1,2}+ \tL_{4,3} \tR_{3,2}+\tL_{4,4} \tR_{4,2}+\tL_{4,5} \tR_{5,2} \\
&\leq \underbrace{\tL_{4,1}}_{{}\le 4\epsilon} + \tL_{4,3} \underbrace{(1-\tR_{1,2})}_{{}\le 0.2}+ \underbrace{\tL_{4,4}}_{{}\le 10\epsilon}+\underbrace{\tR_{5,2}}_{{}\le 50\epsilon},
\end{aligned}
\end{equation*}
which gives us the lower bound $0.346\leq \tL_{4,3}$. 

Similarly, consider the $3$rd column of~$\tW$
and observe that
\begin{equation*}
\begin{aligned}
0.0335 \leq \tW_{3,3}=\tL_{3,:} \cdot \tR_{:,3} &= \tL_{3,1} \tR_{1,3}+ \tL_{3,3} \tR_{3,3}+\tL_{3,4} \tR_{4,3}+\tL_{3,5} \tR_{5,3} \\
&\leq \underbrace{\tL_{3,1}}_{{}\le 2\epsilon} + \max\{\tL_{3,3},\tL_{3,4}\} \underbrace{(1-\tR_{1,3})}_{{}\le 0.714}+\underbrace{\tR_{5,3}}_{{}\le 23\epsilon}. \notag
\end{aligned}
\end{equation*}
The lower bound $\max\{\tL_{3,3},\tL_{3,4}\} \geq 0.0465 $ follows.
\end{proof}

As we have seen above, Lemma~\ref{l:constraints} follows
from Claims~\ref{c:y1+y2} and~\ref{c:o1+o2}.

\section{Conclusions}

In this paper we have solved the Cohen--Rothblum problem,
showing that nonnegative ranks over $\R$ and over $\Q$ may
differ.
More precisely, our construction applies to
matrices of rank~$4$ and greater.
It was already known to Cohen and Rothblum~\cite{CohenRothblum93} that
nonnegative ranks over $\R$ and $\Q$ coincide for matrices of rank at most~$2$,
and Kubjas et al.~\cite{kubjas2015fixed} showed that this also holds for
matrices of nonnegative rank (over~$\R$) at most~$3$. The remaining open
question is whether nonnegative ranks over $\R$ and over $\Q$ differ for
rank-$3$ matrices whose nonnegative rank (over~$\R$) is at least~$4$---%
or whether our example is optimal in this sense.

As our results show that the nonnegative ranks over~\R and~\Q
are different functions, the computability question emerges.
It has long been known (see, e.g., Cohen and Rothblum~\cite{CohenRothblum93}) that the nonnegative rank over~\R is computable,
via a reduction to the existential theory of the reals, which in turn can be decided in~\PSPACE.
(Recently, Shitov has proposed a reduction in the converse direction, i.e., from the existential theory of the reals to NMF~\cite{ShitovUniversality}.)
In contrast, it is not known whether the nonnegative rank over~\Q is computable.
While there is a natural reduction to the decision problem for the
existential theory of the rationals,
the decidability of the latter is
a long-standing and very prominent open question~\cite{poonen2003hilbert}.

Finally, we would like to point out that the complexity of
the following geometric problem closely linked to NMF,
the \emph{nested polytope} problem, is not fully known.
This problem asks, given an ordered field $\mathbb F$ and polytopes
$\mathcal S \sset \mathcal T$ in $\mathbb F^n$,
whether there exists a \emph{simple} polytope~$\mathcal N$
such that $\mathcal S \sset \mathcal N \sset \mathcal T$
(cf.~Gillis and~Glineur~\cite{gillis2012geometric}).
The definition of ``simple'' can be
either ``having at most $k$ vertices,''
or ``having at most $k$ facets,''
or a combination of both.
For $\mathbb F = \R$,
minimizing the number of vertices or, dually, facets
is known to require irrational numbers~\cite{16CKMSW-ICALP}
even in the case of full-dimensional~$\mathcal S$.
While for some representations of the polytopes
such questions are known
to be \NP-hard (see, e.g., Das and Goodrich~\cite{Das1995}),
their precise complexity is not known in general.

\section*{Acknowledgements}
The authors would like to thank Vladimir Lysikov and Vladimir Shiryaev
for stimulating discussions.
Dmitry Chistikov was sponsored in part by the ERC Synergy award ImPACT and is supported by the ERC grant AVS-ISS (648701).
Stefan Kiefer is supported by a University Research Fellowship of the Royal Society.
Stefan Kiefer, Ines Maru\v{s}i\'{c}, Mahsa Shirmohammadi, and James Worrell gratefully acknowledge the support of the EPSRC.


\bibliographystyle{siamplain}
\bibliography{references}

\appendix
\section{Uniqueness of type-1 NMFs of~$M$}\label{app:uniqueness}
In this section, we strengthen
Proposition~\ref{prop:types}~(1)
to show that any type-1 NMF of~$M$ coincides with the one given in Equation~\eqref{eq:fact}, up to a permutation of the columns of~$W$ and the rows of~$\begin{pmatrix} H^{\prime} & H_\eps \end{pmatrix}$.
Together with the other parts of Proposition~\ref{prop:types},
this implies that the NMF~\eqref{eq:fact} is the only
$5$-dimensional stochastic NMF of the matrix~$M$,
up to permutations.

\begin{proposition} \label{prop:NMFtypes}
If $M = L \cdot R$ is a type-1 NMF then $L$ is equal to $W$ up to a permutation of its columns. 
\end{proposition}
\begin{proof}
We recall from Figure~\ref{fig-concave-display} that the supporting polygon~$\mathcal{S}_{q_1^*}$, nested between the triangle  $\triangle r_1 r_2 r_3$ and the polygon $\mathcal{P}_0$, is the triangle $\triangle q_1^* q_3^* q_2^*$. Similarly, as seen in Figure~\ref{xz-plane}, the supporting polygon~$\mathcal{S}_{q_1^*}$, nested between the triangle $\triangle r_4 r_5 r_6$ and the polygon~$\mathcal{P}_1$, is the triangle $\triangle q_1^* q_5^* q_4^*$. We have already shown that $q_1= q_1^*$. In the following, we show that $q_i = q_i^*$ for each $i \in \{2, 3, 4, 5\}$.

Towards a contradiction, suppose that $q_2 \neq q_2^*$ or $q_3 \neq q_3^*$. Let us consider the case when $q_2 \neq q_2^*$. Observe that triangles $\triangle q_1^* q_3^* q_2^*$ and $\triangle q_1^* q_3 q_2$ are both nested between $\triangle r_1 r_2 r_3$ and~$\mathcal{P}_0$.
The fact that $\triangle r_1 r_2 r_3 \subseteq \triangle q_1^* q_3 q_2$ implies that vertices $q_3$ and $q_2$ lie to the right of (or on) directed line segments $q_1^* q_3^*$ and $q_2^* q_1^*$, respectively. Since, moreover, $q_3, q_2 \in \mathcal{P}_0$, it holds that vertex $q_3$ lies to the left of (or on) directed line segment $q_3^* q_2^*$, whereas vertex $q_2$ lies strictly to the left of $q_3^* q_2^*$. However, this implies that the point $r_2$ is to the right of directed line segment $q_3 q_2$, which is a contradiction with the assumption that $\triangle r_1 r_2 r_3 \subseteq \triangle q_1^* q_3 q_2$.
%
%
%
The case $q_3 \neq q_3^*$ analogously leads to a contradiction.
We conclude that $q_2 = q_2^*$ and $q_3 = q_3^*$. Analogously, using Lemma~\ref{lem:ineq2} one can show that $q_4 = q_4^*$ and $q_5 = q_5^*$.

Since $f(q_i) = L_{:,i}$ and $f(q_i^*) = W_{:,i}$ for each $i \in \{2, 3, 4, 5\}$, we conclude that $\{L_{:,2},L_{:,3}\} = \{W_{:,2}, W_{:,3}\}$ and $\{L_{:,4}, L_{:,5}\} = \{W_{:,4}, W_{:,5}\}$.
Therefore, the NMF $M= L \cdot R$ coincides with the one given in~\eqref{eq:fact}, up to a permutation of the columns of~$W$ and the rows of~$\begin{pmatrix} H^{\prime} & H_\eps \end{pmatrix}$. 
\end{proof}

\section{A non-constructive approach to defining $W_\eps$}
\label{app:topological}
Instead of deducing the result of Section~\ref{subsec-Profile3}
from Lemma~\ref{l:constraints},
one can alternatively rely on its weaker form,
Claim~\ref{c:no-4-noeps}, and give
a non-constructive proof of the existence
of an appropriate~$W_\eps$
(satisfying the three desiderata given as bullet points
 on p.~\pageref{desiderata})
via a topological argument that we sketch below.
However, we emphasize that
we do not know how to prove Claim~\ref{c:no-4-noeps}
without following the arguments that prove
Lemma~\ref{l:constraints}.

\begin{proposition}
\label{prop:topological}
There exists a $6 \times 5$ matrix such that
\begin{itemize}[leftmargin=*]
\item its entries are rational,
\item its columns belong to the convex hull of the columns of $W$, and
\item it has no type-4 NMF.
\end{itemize}
\end{proposition}

\begin{proof}
We first employ the geometric constructions
of Section~\ref{sub:parametrization} to argue that
every neighbourhood of the matrix~$W$ contains
a rational matrix that factors through~$W$,
i.e., whose columns belong to the convex hull
 of the columns of~$W$.
Indeed, consider the set $\mathcal F$ of all
stochastic real matrices of size $6 \times 5$ that have
a stochastic NMF with left factor~$W$.
Observe that $\mathcal F$ can be characterized
as the set of matrices whose columns lie in the image
under~$f$
of a full-dimensional set in~$\R^3$, namely of
the convex hull of~$q_1^*, \ldots, q_5^*$.
Since the map~$f$ is specified by matrices~$C$ and~$d$ with
rational coefficients, it immediately follows
that the set of rational matrices is dense in~$\mathcal F$.
As every $\delta$-neighbourhood of the matrix~$W$ includes
some $3$-dimensional subset of~$\mathcal F$,
it also contains
a rational matrix~$W_\delta$ from~$\mathcal F$,
as we wished to prove.


Now assume for the sake of contradiction that
every rational matrix in the set~$\mathcal F$ has a type-4 NMF.
Then the matrices~$W_\delta$ from above also have type-4 NMFs
$W_\delta = L_\delta \cdot R_\delta$, for all $\delta > 0$.
By compactness, there exists
a subsequence of matrices~$W_\delta$ with decreasing~$\delta$
such that the corresponding sequences~$L_\delta$ and~$R_\delta$
converge. Taking the limit, we arrive
at the equality $W = L \cdot R$ where the right-hand side
is also a type-4 NMF---%
which contradicts Claim~\ref{c:no-4-noeps}.
This completes the proof.
(Note that
Lemma~\ref{l:constraints}
contains a constructive version of this argument.)
\end{proof}

It is worth mentioning that this reasoning
follows similar lines as
the \emph{upper semi-continuity} argument
for nonnegative rank~\cite{BocciCR11}:
the nonnegative rank of any (rational or irrational) matrix~$W_\eps$
which is entry-wise close enough to $W$
can only be larger than or equal to that of~$W$.

\end{document}